%% file: two_sided_markets.tex
\documentclass[11pt, a4paper, UKenglish]{article}
\usepackage{booktabs} % For formal tables
\usepackage[noend,ruled]{algorithm2e} % For algorithms

\SetAlFnt{\small}
\SetAlCapFnt{\small}
\SetAlCapNameFnt{\small}
\SetAlCapHSkip{0pt}
\IncMargin{-\parindent}

\usepackage{hyperref}
\usepackage[utf8]{inputenc}
\usepackage[margin=1in]{geometry}

\usepackage{dsfont}
\usepackage{amsmath,amsthm, amsfonts}%,amssymb}
\usepackage{mathtools}
\usepackage{thmtools, thm-restate}
\newtheorem{lemma}{Lemma}
\newtheorem{theorem}{Theorem}

\newtheorem{proposition}{Proposition}

\theoremstyle{definition}

\usepackage{natbib}
\usepackage{xcolor}
\newcommand{\todo}[1]{}
\newcommand{\later}[1]{}

\newcommand{\OPT}{\mathrm{OPT}}
\newcommand{\ALG}{\mathrm{ALG}}

\newcommand{\Ex}[2][]{\mbox{\rm\bf E}_{#1}\left[#2\right]}
\renewcommand{\Pr}[2][]{\mbox{\rm\bf Pr}_{#1}\left[#2\right]}
\newcommand{\e}{\mathrm{e}}

\newcommand{\Msell}{M_{\textnormal{SELL}}}
\newcommand{\Mbuy}{M_{\textnormal{BUY}}}

\newcommand{\A}{\mathcal{A}}
\newcommand{\SW}{\textnormal{SW}}

\title{Truthful Mechanisms for Two-Sided Markets via Prophet Inequalities}

\author{Alexander Braun \footnote{ alexander.braun@uni-bonn.de} \qquad Thomas Kesselheim \footnote{ thomas.kesselheim@uni-bonn.de} \\{\small Institute of Computer Science, University of Bonn, Germany}}

\begin{document}
	\maketitle
	\begin{abstract}
		\input{abstract}
	\end{abstract}
	\newpage
	\input{introduction}
	\input{preliminaries}
	\input{matroids_sbb_offline}

	\input{matroids_wbb}	
	\input{xos_sbb}	
	\input{knapsack_sbb}

	\input{knapsack_wbb}
	\input{conclusion}
	\bibliography{references,dblp}
	\newpage
	\appendix
	\input{appendix_bilateral_trade}
	\newpage
	\input{appendix_matroids_sbb_offline}

	\newpage
	\input{appendix_matroids_wbb}
	\newpage
	\input{appendix_xos_sbb} 
	\newpage
	\input{appendix_knapsack_wbb}

\end{document}

%% file: abstract.tex
We design novel mechanisms for welfare-maximization in two-sided markets. That is, there are buyers willing to purchase items and sellers holding items initially, both acting rationally and strategically in order to maximize utility. Our mechanisms are designed based on a powerful correspondence between two-sided markets and prophet inequalities. They satisfy individual rationality, dominant-strategy incentive compatibility, budget-balance constraints and give constant-factor approximations to the optimal social welfare.

We improve previous results in several settings: Our main focus is on matroid double auctions, where the set of buyers who obtain an item needs to be independent in a matroid. We construct two mechanisms, the first being a $1/3$-approximation of the optimal social welfare satisfying strong budget-balance and requiring the agents to trade in a customized order, the second being a $1/2$-approximation, weakly budget-balanced and able to deal with online arrival determined by an adversary. In addition, we construct constant-factor approximations in two-sided markets when buyers need to fulfill a knapsack constraint. Also, in combinatorial double auctions, where buyers have valuation functions over item bundles instead of being interested in only one item, using similar techniques, we design a mechanism which is a $1/2$-approximation of the optimal social welfare, strongly budget-balanced and can deal with online arrival of agents in an adversarial order.

%% file: introduction.tex
\section{Introduction}

Mechanisms for allocation problems in one-sided markets have been studied for decades. The goal is to allocate items to agents in order to maximize either the revenue of the auctioneer or the social welfare. In this setting, a fundamental assumption is that all items are initially held by the auctioneer who does not have any value for any of them. Various different auctions and allocation procedures for one-sided markets have been developed, such as VCG \citep{RePEc:ecm:emetrp:v:41:y:1973:i:4:p:617-31, RePEc:bla:jfinan:v:16:y:1961:i:1:p:8-37, RePEc:kap:pubcho:v:11:y:1971:i:1:p:17-33}, posted-prices mechanisms (as e.g. in \citet{DBLP:conf/stoc/ChawlaHMS10, DBLP:conf/focs/DuettingFKL17}) and many other auction formats.

In a related but different setting, items are held by strategic sellers initially.
%Over time, the question arose how the allocation problem changes if there are sellers holding the items initially.
That is, each seller has a valuation over her bundle of items and acts rationally %and strategically 
with the goal to maximize utility. Examples are widely spread, as to mention stock exchanges, ad auctions or online marketplaces such as ebay. In the context of mechanism design, this imposes the following task: construct a mechanism which specifies trades between buyers and sellers and determine suitable prices for each trade with the objective of maximizing the overall social welfare.

Standard requirements for mechanisms are \emph{individual rationality} (IR) and \emph{dominant strategy incentive compatibility} (DSIC). The former means that it cannot be harmful for any agent to participate in the mechanism---the latter that reporting preferences truthfully is a dominant strategy for any agent, no matter what other agents report. Furthermore, as one cannot assume that there is a superior authority funding beneficial trades in two-sided markets, an additional natural requirement is \emph{budget balance}. Its stronger version, \emph{strong budget balance} (SBB), means that the mechanism can neither subsidize trades nor is allowed to extract money from trades. In other words, this requires that all money which is spent by buyers is transferred to sellers. The weaker form, \emph{weak budget balance} (WBB), only requires the first property, namely that subsidizing trades is prohibited, but the mechanism is allowed to extract money from trades.

Unfortunately, in their seminal work from 1983, \citet{5048d8a1d97645ff80d34f037fa4971a} showed that no mechanism can simultaneously be individually rational, incentive compatible, budget balanced and optimize social welfare\footnote{The original result from Myerson and Satterthwaite is for bilateral trade instances, i.e. one seller holding one item and one buyer. They show that even individual rationality and (Bayesian) incentive compatibility cannot be combined with achieving the optimal ex-post social welfare.}. This result is a sharp contrast to one-sided markets where optimal results are possible \citep{RePEc:bla:jfinan:v:16:y:1961:i:1:p:8-37, DBLP:journals/mor/Myerson81}. As a consequence, \emph{approximating} the optimal social welfare becomes a natural challenge. Even further, trying to approximate the optimal social welfare with a rather simple mechanism which can be easily understood by all participants may be an even more desirable goal.

Probably the most fundamental problem is \emph{bilateral trade}  (see e.g. \citet{5048d8a1d97645ff80d34f037fa4971a, DBLP:journals/corr/BlumrosenD16, Kang2018StrategyProofAO}). There is one seller holding one indivisible item and one buyer. In more general \emph{double auctions}, there might be multiple buyers, multiple sellers, multiple items, and complex combinatorial constraints. In \emph{matroid double auctions} (see e.g. \citet{10.1145/2600057.2602854, 10.5555/2884435.2884533}), each seller initially holds one of $m$ identical items, each buyer wants to purchase at most one of them and the set of buyers who receive an item needs to be an independent set in a matroid. In \emph{combinatorial double auctions}, there are $k$ sellers holding $m$ heterogeneous items and the agents have combinatorial valuation functions over item bundles (see e.g. \citet{10.1145/2600057.2602843, 10.1145/3381523}). In \emph{knapsack double auctions}, the matroid constraint over the set of buyers is replaced by a knapsack constraint (see e.g. \citet{10.1145/2600057.2602854}). That is, each buyer has a weight and we need to select buyers in a way such that the sum of weights does not exceed a certain capacity. In any of the settings, agents are assumed to maximize their (quasi-linear) utilities. It is a standard assumption in this context to assume a Bayesian setting: All agents have privately known valuation functions over item bundles. These valuation functions are drawn independently from (possibly different) publicly known distributions.

\subsection{Bilateral Trade via Prophet Inequalities}
For bilateral trade, there is a very simple mechanism template: Let $v_s$ denote the seller's value and $v_b$ denote the buyer's value for the item. Both are drawn independently from some probability distributions. Fix a price $p$ and trade the item if and only if $v_b \geq p \geq v_s$. Among others, \citet{10.1145/2600057.2602843, DBLP:journals/corr/BlumrosenD16} and \citet{Gerstgrasser_Goldberg_deKeijzer_Lazos_Skopalik_2019} set $p$ to be the median of the seller's distribution which recovers an expected welfare of at least $\frac{1}{2} \cdot \Ex[]{\max \{ v_s, v_b \}}$, so it is a $\frac{1}{2}$-approximation.

This simple mechanism can also be interpreted as a sequential posted-prices mechanism with price $p$, the mechanism first asks seller $s$ if she would like to keep or try selling the item for price $p$. Afterwards, buyer $b$ may purchase the item for price $p$ if the seller accepted a trade. This setting is conceptually similar to posting a price $p$ in a one-sided market with two buyers. The only difference is that if both values are below $p$ the seller keeps the item whereas in the one-sided market it is assumed to be discarded. 

Due to this correspondence, we can easily lower-bound the social welfare via \emph{prophet inequalities}, originally introduced by \citet{krengel1977semiamarts, krengel1978semiamarts} and \citet{Samuel84}. In particular, one can also recover half of the optimal social welfare as follows: Instead of applying a pricing strategy via quantiles, one could also use a \emph{balanced price} $p = \frac{1}{2} \cdot \Ex[]{\max \{ v_s, v_b \}}$. A proof of the approximation guarantee directly follows by the above considerations (see Appendix~\ref{appendix:bilateral_trade} for details) via standard prophet inequality results \citep{10.1145/2213977.2213991, DBLP:conf/soda/FeldmanGL15} for $n = 2$. 

The observation immediately raises the question whether all problems in two-sided markets can be solved via posted-prices mechanisms and prophet inequalities in such a straightforward way. We might interpret sellers as buyers, consider them first and ask which items they would like to keep, afterwards offer the remaining items to buyers. Unfortunately, this brings about a number of issues. First and foremost, budget balance is not guaranteed because the payments by the buyers will usually not match what we promise the sellers.

In this paper, we demonstrate that nonetheless the pricing and proof strategies from prophet inequalities give us powerful tools to design mechanisms in two-sided markets. With these, we are able to design mechanisms which obtain improved approximation guarantees concerning social welfare.

\subsection{Our Results}

Based on the technique of balanced prices \citep{10.1145/2213977.2213991,DBLP:conf/focs/DuettingFKL17}, we design mechanisms for two-sided markets. All our mechanisms are DSIC and IR for all agents; they fulfill different variants of budget balance.

Our main results are two mechanisms for double auctions with a matroid constraint over the set of buyers. The first mechanism (Section~\ref{section:matroide_sbb}) is strongly budget-balanced and a $1/3$-approximation with respect to the optimal social welfare. It relies on balanced prices used to obtain matroid prophet inequalities \citep{10.1145/2213977.2213991}. As these prices change based on previous decisions, our mechanism has to carefully choose the order in which agents are offered trades.

The second mechanism (Section~\ref{section:matroide_wbb}) is only weakly budget-balanced but the approximation guarantee improves to a $1/2$-fraction of the optimal social welfare. Another advantage is that the order can be arbitrary. That is, this mechanism can also deal with online arrival of agents where the order can be chosen by an adversary. The adversary may even adapt the order depending on the agents and realizations before. The best results for this setting so far were approximation ratios of $1/16$ by \citet{10.5555/2884435.2884533} and $1/(3 + \sqrt{3})$ by \citet{duetting2020efficient}\footnote{We note that the setting in \citet{duetting2020efficient} is different from ours as they construct mechanisms with limited sample-based knowledge of the distributions. Still, their results were the best known so far, also for the setting with complete knowledge of the distributions.}.

Furthermore, we develop a mechanism for combinatorial double auctions (Section~\ref{section:xos_sbb}), which is derived from the result by \citet{DBLP:conf/soda/FeldmanGL15} in one-sided markets. Our mechanism is truthful in two different settings. On the one hand, we assume that each seller holds one item initially and that buyers have fractionally subadditive valuation functions over item bundles. On the other hand, we can allow sellers to hold multiple heterogeneous items and all agents to have additive valuation functions. The mechanism is DSIC, IR, strongly budget-balanced and can handle online adversarial arrivals of buyers. The mechanism's approximation ratio is $1/2$ with respect to the optimal social welfare. By this, we improve the factor of $1/6$ in \citet{10.1145/3381523} and the bounds obtained by \citet{duetting2020efficient}, who are able to show a bound of $1/3$ for unit-supply sellers and buyers with fractionally subadditive valuation functions. In contrast to our results, they require agents to be allowed to trade in any order whereas our results for combinatorial double auctions still hold when agents arrive online and the order is determined by an adversary.

Finally, we can design two mechanisms for knapsack double auctions where the set of buyers who receive an item needs to be feasible with respect to a knapsack constraint, sellers bring identical items to the market and buyers have unit-demand valuations. We obtain a strongly budget-balanced mechanism in Section \ref{section:knapsack_sbb}, which considers agents online in a customized order leading to a competitive ratio of $1/10$. On the other hand, we improve this guarantee to $1/7$ by a weakly budget-balanced mechanism for online adversarial arrival order in Section \ref{section:knapsack_wbb}.

\begin{table*}[h]
		\begin{tabular}{|ll|l|l|l||l|}
			\hline
			&   & \textbf{Budget-Bal.} & \textbf{Approx.} & \textbf{Trading Order} & \textbf{Previous Best} \\
			\hline 	
			\multicolumn{2}{|l|}{\emph{Matroid DA: }}  & Strong & $1/3$  & Offline & $1/16$\textsuperscript{ a} and  \\
			& &  Weak & $1/2$  & Online Adv. &  $1/(3+\sqrt{3})$\textsuperscript{ b}\\
			\hline
			\multicolumn{2}{|l|}{\emph{Combinatorial DA: }} & & & & \\
			& XOS + Unit-Supply & Strong &  $1/2$ & Online Adv. & $1/6$\textsuperscript{ c} and  \\
			& Additive + Additive  &  Strong &  $1/2$ &Online Adv. & $1/3$\textsuperscript{ b} \\
			\hline
			\multicolumn{2}{|l|}{\emph{Knapsack DA: }} & Strong & $1/10$ &  Online Custom. & \\
			&  & Weak & $1/7$ & Online Adv.&  \\
			\hline
		\end{tabular} 
		\label{Table:Results}
		\caption{Our state-of-the-art approximation guarantees for mechanisms in matroid, combinatorial and knapsack double auctions. Concerning the previous best results, 'a' can be found in \citet{10.5555/2884435.2884533}, 'b' in \citet{duetting2020efficient} and 'c' in \citet{10.1145/3381523}.}
\end{table*}

Observe that for the settings with online adversarial arrival order of buyers, the approximation ratio of $1/2$ (as in matroid double auctions and combinatorial double auctions) matches a tight upper bound as soon as there are at least two buyers, which corresponds to the commonly known instance for prophet inequalities (see Section~\ref{section:upper_bounds}). This, of course, does not apply to bilateral trades, for which deriving the optimal approximation ratio is still an open problem.

As a side remark, when only having sample-based access instead of full knowledge of the distribution, in one-sided environments it is known that $Poly(n,m,\epsilon^{-1})$ samples from every distribution suffice to lose only an additive $O(n\epsilon)$-term \citep{DBLP:conf/focs/DuettingFKL17}. These guarantees for balanced pricing carry over to the two-sided market setting. For details concerning the techniques on sample-based access, we refer to the respective sections in \citet{DBLP:conf/soda/FeldmanGL15} and \citet{DBLP:conf/focs/DuettingFKL17}.

\subsection{Our Techniques}

The idea behind balanced prices in one-sided markets \citep{10.1145/2213977.2213991, DBLP:conf/soda/FeldmanGL15, DBLP:conf/focs/DuettingFKL17} is that they are low enough so that the agents can afford the items they are allocated in the social optimum (which means they have high utility). At the same time, they should be high enough so that the revenue covers the loss in social welfare due to allocations not in line with the social optimum.

In two-sided markets, in order to propose trades we use prices that are again low enough and high enough. On the one hand, agents may have values exceeding the prices and hence, either keep items (as sellers) or purchase items (as buyers). On the other hand, prices should be high enough so that once an agent keeps or buys an item, we can ensure that her value is sufficiently large to cover the loss in social welfare by allocating the item.

Our proofs concerning the approximation guarantees mimic the spirit of revenue and utility based-ones in one-sided markets: we split the contribution to welfare of each agent into the base value, defined by the price of the proposed trade, and surplus, which is the amount by how much the agent's value exceeds the price. %, i.e. we cut the value of an agent for an item into two pieces.
Afterwards, we bound each quantity separately. As a matter of fact, it does not play a key role which agent purchases or keeps which item---since any irrevocably allocated item ensures a sufficient contribution to welfare via its price. This is a sharp contrast to mechanisms in which the output allocation plays a key role in order to obtain a specific approximation guarantee.

Sequential posted-prices mechanisms in one-sided markets are DSIC by design as we only offer any agent the possibility to purchase items at most once. When extending these concepts to two-sided markets, an additional major challenge is to fulfill the budget-balance constraint. As illustrated above, a straightforward generalization of the prophet inequality techniques might not be possible as it may lack money in trades. To overcome these problems, one would like to consider agents multiple times for trades. Nonetheless, truthfulness constraints may get violated if we offer one agent more than one trade. 

In order to obtain both, budget-balance and dominant strategy incentive compatibility, our mechanisms carefully choose when deciding to propose which trade between which agents at which price. In particular, if we offer different prices to one seller, they should not increase over time because otherwise the seller might strategically wait for a higher price and by this, reject earlier trades.

\subsection{Paper Organization}

The paper is structured as follows: In Section \ref{Section:Preliminaries}, we give basic definitions and notation. In Section \ref{section:matroide_sbb}, we state our mechanism for matroid double auctions which satisfies strong budget balance with a proof for the approximation guarantee in the simplified full information setting. The proof for the general incomplete information setting is given in Appendix~\ref{appendix:matroids_sbb}. In Section~\ref{section:matroide_wbb}, we give our mechanism for matroid double auctions which satisfies weak budget balance for any online adversarial arrival order of agents. The proof for the competitive ratio can be found in Appendix~\ref{appendix:matroid_wbb}. In Section \ref{section:xos_sbb}, we consider the case of combinatorial double auctions, state our mechanism and give a high-level reduction to prophet inequalities. The formal proof of the competitive ratio which is highly related to the proof of the corresponding prophet inequality for XOS-valuation functions can be found in Appendix \ref{appendix:xos_sbb}. In Section \ref{section:knapsack_sbb}, we consider knapsack double auctions with strong budget balance and give a proof for the approximation guarantee. In Section \ref{section:knapsack_wbb}, we state our mechanism for knapsack double auctions which admits online adversarial arrival of agents and is weakly budget-balanced. The proof for the competitive ratio can be found in Appendix \ref{appendix:knapsack_wbb}. Final remarks and questions for future research can be found in Section~\ref{Section:Conclusion}. 
 
\input{related_work}

%% file: related_work.tex
\subsection{Related Work}

Two-sided markets have been studied for a long time, including the mentioned impossibility result by \citet{5048d8a1d97645ff80d34f037fa4971a} and pioneering work on trade-reduction mechanisms and their generalizations as e.g. considered in \citet{MCAFEE1992434, 10.1145/2600057.2602854, 10.1145/779928.779937, 10.5555/1622467.1622485}. Only much more recently, worst-case approximation ratios have been considered. There has been a lot of progress on improving the guarantees for bilateral trade \citep[among others:][]{10.1145/2600057.2602843, DBLP:journals/corr/BlumrosenD16, Kang2018StrategyProofAO, Gerstgrasser_Goldberg_deKeijzer_Lazos_Skopalik_2019}. However, determining the optimal guarantee is still an open problem.

Most relevant to our work are the ones of \citet{10.5555/2884435.2884533} and \citet{10.1145/3381523}, which derive mechanisms for matroid and combinatorial double auctions in Bayesian settings. \citet{10.5555/2884435.2884533} focus on matroid double auctions, designing mechanisms with pricing strategies based on quantiles, whereas our approach uses balanced prices.
%. As illustrated above, in contrast our pricing is based on balanced prices.
\citet{10.1145/3381523} consider combinatorial double auctions using very similar prices as ours. However, the analysis is different as their proofs rely on case distinctions where our proofs use charging arguments from balanced prices. Another important contribution of \citet{10.1145/3381523} is the introduction and discussion of direct-trade budget-balance, which we also adopt in this paper.
% in . Compared to their mechanisms for, our mechanisms from Section \ref{section:xos_sbb} are simpler and obtain better approximation guarantees. In addition, . 
\citet{duetting2020efficient} consider the same constraints. Besides giving improved approximation guarantees, they change the fundamental assumption of the Bayesian setting: They design mechanisms given only sample-based access to the underlying distribution.

There is also a line of work using different objective functions in two-sided markets, most prominently \emph{gain from trade} \citep{10.1007/978-3-662-54110-4_28, 10.1145/3033274.3085148, colini-baldeschi2017a, 10.1145/3219166.3219203, Segal_Halevi_2016, 10.1007/978-3-319-99660-8_15, DBLP:journals/corr/Segal-HaleviHA16, cai2020multidimensional}. In this setting, only the \emph{increase} in welfare by transferring items from sellers to buyers is measured. An $\alpha$-approximation with respect to gain from trade is also an $\alpha$-approximation with respect to social welfare but not vice versa. Indeed, \citet{DBLP:journals/corr/BlumrosenD16} and \citet{10.1007/978-3-662-54110-4_28} show that approximating the gain from trade is harder than social welfare: There is no DISC, IR and SBB mechanism which can achieve a constant factor approximation to the optimal gain from trade. \citet{DBLP:conf/soda/BabaioffGG20} tackle the question by how many buyers and sellers the size of the two-sided market needs to be increased in order to recover the optimal gain from trade from the original market, mirroring the seminal work of \citet{RePEc:aea:aecrev:v:86:y:1996:i:1:p:180-94}. Another interesting objective function is the profit of the sellers in two-sided markets as considered by \citet{DBLP:journals/corr/abs-1906-09305}.

Prophet inequalities date back even to the 1970s \citep{krengel1977semiamarts, krengel1978semiamarts,Samuel84}. \citet{DBLP:conf/aaai/HajiaghayiKS07} and \citet{DBLP:conf/stoc/ChawlaHMS10} introduced their use in algorithmic mechanism design for one-sided markets. They mainly have two applications: On the one hand, they can be interpreted as posted-price mechanisms for welfare maximization with multiple buyers. On the other hand, they provide a useful tool to approximate revenue for one buyer with multiple items. Most relevant to our approach is the concept of balanced prices as applied by \citet{10.1145/2213977.2213991} and \citet{DBLP:conf/esa/DuttingK15} for settings with matroid constraints, in \citet{DBLP:conf/soda/FeldmanGL15} for combinatorial auctions and in \citet{DBLP:conf/focs/DuettingFKL17} as a generalized version of both. For a more detailed overview on prophet inequalities in the context of posted-prices mechanisms, we refer to Lucier's excellent survey \citep{DBLP:journals/sigecom/Lucier17}.

%% file: preliminaries.tex
\section{Preliminaries}
\label{Section:Preliminaries}

We consider the following setup for two-sided markets: There is a set of $n$ buyers $B$, a set of $k$ sellers $S$ and a set of $m$ items $M$. We assume that $B \cap S = \emptyset$, so any agent can either act as a buyer or a seller. Before running any (reallocation) mechanism, the set of items is initially held by the sellers. We denote by $I_l$ the set of items which is hold by seller $l$ initially and call the vector $\left( I_1,\dots, I_{k} \right)$ the \emph{initial allocation}. Note that the sets $I_l$ are pairwise disjoint, i.e. for any two sellers $l\neq l'$ we have $I_l \cap I_{l'} = \emptyset$, and further all items are allocated to some seller before running our mechanism, i.e. $\bigcup_{l \in S} I_l = M$. \\ 
Any agent $i \in B \cup S$ has a privately known valuation function $v_i : 2^M \rightarrow \mathbb{R}_{\geq 0}$. For $T \subseteq M$, we denote by $v_i(T)$ the value of agent $i$ for being allocated item bundle $T$. Any seller $l$ is assumed to have only positive value for items in her initial bundle $I_l$, i.e. for any seller $l \in S$ and $T \subseteq M$ it holds that $v_l(T) = v_l(T \cap I_l)$. Valuation functions are always non-negative and bounded for any bundle as well as monotone and normalized, i.e. $v_i(T) \leq v_i(T')$ for $T \subseteq T' \subseteq M$ and $v_i(\emptyset) = 0$. 
We consider a Bayesian setting where each agent $i$'s valuation function is drawn independently from a publicly known, not necessarily identical probability distribution $\mathcal{D}_i$, that is, $\mathcal{D}_i$ is a probability distribution over the \textit{space of valuation functions} $V_i$. We denote by $\mathcal{D} = \times_{i \in B \cup S} \mathcal{D}_i$ the joint probability distribution of the space of all agents' valuation functions $\mathbf{V} = \times_{i \in B \cup S} V_i$ and refer to $\mathbf{v}$ as a \textit{valuation profile} which consists of one valuation function per agent. \\
An \textit{allocation} $\mathbf{X} = (X_i)_{i \in B \cup S}$ is a vector of item bundles such that agent $i$ is allocated bundle $X_i$ and for two agents $i \neq i'$, we have $X_i \cap X_{i'} = \emptyset$. The \textit{social welfare} of an allocation $\mathbf{X}$ given valuation profile $\mathbf{v}$ is defined as $\mathbf{v}(\mathbf{X}) \coloneqq \sum_{i \in B \cup S} v_i(X_i)$. Concerning \textit{feasibility}, as said, any seller $l \in S$ can only receive items in her initial allocation, i.e. $X_l \subseteq I_l$ for any $l \in S$.

\subsection*{Mechanisms and their properties}

A (direct revelation) \textit{mechanism} takes as input a vector of valuation functions which are reported by agents. Agents can report any possible valuation in their space of valuation functions $V_i$, not necessarily their true one. A mechanism outputs an allocation of items to agents $\mathbf{X}$ as well as payments $\mathbf{P}$. For buyers, payments are negative meaning that they pay money to the mechanism whereas for sellers, payments are positive as they receive money. 

Agents are assumed to maximize \textit{utility}. Fixing a valuation profile $\mathbf{v}$, an allocation $\mathbf{X}$ and payments $\mathbf{P}$, the (quasi-linear) utility of buyer $i$ for being allocated bundle $X_i \subseteq M$ is given by $u_i(X_i) = v_i(X_i) - P_i$ whereas the utility for seller $l$ who remains with bundle $X_l \subseteq I_l$ is given by $u_l(X_l) = v_l(X_l) + P_l$.

Our mechanisms are designed to fulfill the following desirable constraints:
\begin{itemize}
	\item \emph{Dominant Strategy Incentive Compatibility} (DSIC): It is a dominant strategy for every agent to report her true valuation independent of the other agents' behavior. 
	\item \emph{Individual Rationality} (IR): When playing this dominant strategy, no agent decreases her utility by participating in the mechanism. So, for buyers $v_i(X_i) - P_i \geq 0$ and for sellers $v_l(X_l) + P_l \geq v_l(I_l)$.
	\item \emph{Weak/Strong Budget Balance} (WBB/SBB): The money received by sellers is at most/equals the payments made by buyers, i.e. $\sum_{i \in B} P_i \stackrel{(=)}{\geq}  \sum_{l \in S} P_l$.
\end{itemize}
Concerning budget-balance, \citet{10.1145/3381523} showed a weakness in the original definition of strong budget-balance as cross subsidizing trades with already received money is not prohibited as long as the sum of payments is equal for buyers and sellers\footnote{\citet{10.1145/3381523} argue that turning a WBB into an SBB mechanism is rather easy with a small loss in the approximation guarantee as one can simply draw one seller uniformly at random and give all the surplus money in the WBB mechanism to this seller.}. The stronger notion of \textit{direct-trade weak/strong budget balance} (DWBB/DSBB) requires that the outcome of the mechanism can be obtained by a composition of bilateral trades where in each trade an item is reallocated from seller $l$ to buyer $i$, payments are transferred from some buyer $i$ to some seller $l$ and each item may only be traded at most once \citep{10.1145/3381523}. If the buyer's payment exceeds the seller's receiving, the mechanism is DWBB, if payments in each of these bilateral trades are equal for buyers and sellers, we refer to DSBB. \\
A truthful mechanism outputting allocation $\mathbf{X}$ is an \textit{$\alpha$-approximation} to the optimal social welfare if $\Ex[\mathbf{v}]{\mathbf{v}(\mathbf{X}) } \geq \alpha \cdot \Ex[\mathbf{v}]{\max_{\mathbf{X}^\ast} \mathbf{v}(\mathbf{X}^\ast)}$. In the case of online arrival of agents, we may use $\alpha$-approximation and \emph{$\alpha$-competitive} interchangeably.

%% file: matroids_sbb_offline.tex
\section{Matroid Double Auctions and Strong Budget-Balance}
\label{section:matroide_sbb}

Our first mechanism is for double auctions where the set of $n$ buyers $B$ is equipped with a matroid\footnote{A \emph{matroid} $\mathcal{M} = \left( \Lambda, \mathcal{I}\right)$ over ground set $\Lambda$ with non-empty set system $\mathcal{I} \subseteq 2^\Lambda$ is defined via the following properties. For two subsets $ X \subseteq Y$ of $\Lambda$ with $Y \in \mathcal{I}$, also $X \in \mathcal{I}$. And for $X,Y \in \mathcal{I}$ with $|X| < |Y|$ there is a $y \in Y \setminus X$ such that $X \cup \{y\} \in \mathcal{I}$. We call sets in $\mathcal{I}$ \emph{independent}.} constraint. That is, there is a matroid $\mathcal{M}_B = \left( B, \mathcal{I}_B \right)$ and the set of buyers who receive an item in the mechanism needs to be an independent set in the matroid $\mathcal{M}_B$. For this section, we assume buyers to be unit-demand\footnote{A valuation function is called \textit{unit-demand} if $v_i(T) = \max_{j \in T} v_i(\{j\})$.} and sellers to be unit-supply, i.e. every seller initially holds a single, indivisible item and hence $k = m$. Items are identical, meaning that $v_i(T)$ only depends on the size of $T$, so agents only care if they get an item or not. Our mechanism requires an offline setting in which buyers and sellers can trade in any order which will be determined during the mechanism. In particular, we assume that we can pick one buyer and one seller in any step out of the remaining ones and offer a trade at some price to both agents. Further, we simplify notation in this section. As the valuation function of any agent boils down to a single value that the agent has for being allocated an item, a valuation profile $\mathbf{v}$ can now be interpreted as a $|B \cup S|$-dimensional vector over the non-negative real numbers in which each entry corresponds to the value of an agent for being allocated an item. We denote this value by $v_i$. Further, as sellers are unit-supply, there is a one-to-one correspondence between sellers and items allowing to denote a seller as well as the corresponding item by $j$.

\subsection*{The Mechanism}

Our mechanism is stated in Algorithm \ref{Matroid_mechanism_sbb} and formally described below. The intuition behind is as follows: We consider a relaxation of the expected optimal social welfare to $\Ex[\mathbf{v}]{ \mathbf{v} \left( \OPT_B(\mathbf{v} ) \right)} + \Ex[\mathbf{v}]{  \sum_{j \in S } v_{j} }$, where $\OPT_B(\mathbf{v})$ is the optimal choice when restricting to the set of buyers. That is, in our relaxation, each item can be counted twice: It will contribute to the first term by being assigned to a buyer while in the second term it is assumed that the seller keeps it. Now, the pricing for a trade between buyer $i$ and seller $j$ needs to ensure that the mechanism recovers a suitable fraction in both terms. In particular, once a trade occurs, the price for this trade covers the loss of both, the seller and the buyer, in the relaxed optimal social welfare. In addition, the remaining share of the social welfare is covered by the surplus. By the choice of the order in which trades are offered, one can ensure that prices are monotone for a fixed seller which is crucial concerning truthfulness and budget-balance.

\begin{algorithm}
	\SetAlgoNoLine
	\DontPrintSemicolon
	%\KwData{Sellers, indexed $1,\dots,m$ and buyers, indexed $m+1, \dots m+n$}
	\KwResult{Set $A$ of agents to get an item with $A \cap B \in \mathcal{I}_B$ and $|A| = |S|$}
	$A_B \longleftarrow \emptyset$ \qquad 	$A_S \longleftarrow \emptyset$ \qquad $r \longleftarrow |S|$ \qquad $\Msell \longleftarrow S$ \qquad $\Mbuy \longleftarrow B$  \\
	\While{$\Mbuy \neq \emptyset$ \textnormal{and} $\Msell \neq \emptyset$}{
		recompute the thresholds $p_{i}(A_B, r)$ and $p_j(A_B,r)$ with respect to current $A_B$, $r$, $\Msell$ and $\Mbuy$ \\
		$j \in \arg\min_{ j' \in \Msell} p_{j'}(A_B,r)$;  \qquad  $i \in \arg\max_{ i' \in \Mbuy} p_{i'}(A_B,r)$  \\ 
		\If{$A_B \cup \{ i \} \notin \mathcal{I}_B$ \textnormal{or} $| A_B \cup \{i\} | > r$}{
			$\Mbuy \longleftarrow \Mbuy \setminus \{i\}$  \\
			\textbf{go to next iteration}
			}
		$p \longleftarrow p_{i,j}\left( A_B, r \right)$ \\
		\If{$v_j > p$}{
			$A_S \longleftarrow A_S\cup \{j\}$; \qquad $\Msell \longleftarrow \Msell \setminus \{j\}$; \qquad $r \longleftarrow r-1$ \\
			}
		\If{$v_j \leq p$}{
			$\Mbuy \longleftarrow \Mbuy \setminus \{i\}$ \\
			\If{$v_i > p$}{
				$A_B \longleftarrow A_B \cup \{i\}$; \qquad $\Msell \longleftarrow \Msell \setminus \{j\}$ \\
				}
		}
		}
		
		\Return $A := A_B \cup A_S \cup \Msell$ 
		\caption{Mechanism for Matroid Double Auctions with Strong Budget Balance}
		\label{Matroid_mechanism_sbb}
\end{algorithm}

Throughout the algorithm, we maintain a set of agents $A = A_B \cup A_S$ who are irrevocably allocated an item. The set $A_B$ contains all buyers who receive an item, so we require $A_B \in \mathcal{I}_B$. The set $A_S$ contains all sellers who irrevocably keep their item. Additionally, in the set $\Msell$ we store all sellers who may still be considered for a possible trade, meaning that we have neither decided to trade their item nor that they keep it. Analogously, the set $\Mbuy$ denotes the set of buyers who have not been considered for a trade yet. In other words, any agent can be listed in one of three different stages throughout our mechanism: all agents in $\Mbuy$ and $\Msell$ are \emph{pending}, meaning that each of these agent can be considered for a trade. As we offer trades to agents, agents may either remain pending, we may \emph{irrevocably allocate} an item to the agent or the agent may be \emph{irrevocably discarded} for holding an item after the mechanism. \\

We maintain buyer-specific thresholds $p_i$ and seller-specific thresholds $p_j$, which roughly speaking represent how much the two terms in the relaxation of the optimum, $\Ex[\mathbf{v}]{ \mathbf{v} \left( \OPT_B(\mathbf{v} ) \right)}$ and $\Ex[\mathbf{v}]{  \sum_{j \in S } v_{j} }$, are harmed by a trade between $j$ and $i$. The price for a trade between seller $j$ and buyer $i$ will then be defined as $p_{i,j} = \textnormal{constant} \cdot ( p_i + p_j) $. By this, we ensure that once an item is irrevocably allocated, the value of the agent who gets the item is high enough to cover the welfare loss in both terms of the relaxed optimal social welfare. In every iteration, among all buyers $i \in \Mbuy$ that can still be added, we consider the one with the largest threshold $p_i$. We try to match her to the seller $j \in \Msell$ with the smallest threshold\footnote{Break ties arbitrarily, but always in the same way.} $p_j$. To this end, we first ask seller $j$ if she wants to sell or keep her item for a price of $p_{i,j}$. If she wants to keep her item, we remove seller $j$ from the set of available sellers. Otherwise, i.e. if seller $j$ considers selling her item, we ask buyer $i$ if she wants to buy the item for price $p_{i,j}$. If buyer $i$ agrees, the item is transferred from $j$ to $i$, both are removed from the set of available agents, $i$ is irrevocably allocated an item, $j$ is irrevocably discarded for holding an item and $i$ pays $p_{i,j}$ to seller $j$. Else, buyer $i$ is removed from the set of available buyers and irrevocably discarded. Then we move to the next iteration, in which we consider a different pair for trading.
	
\subsection*{The Pricing}

	First, by construction, our mechanism never offers trades to buyers who cannot be feasibly added to $A_B$. Hence, the mechanism ensures that the set of buyers $A_B$ who receive an item in our mechanism is an independent set in the matroid, i.e. $A_B \in \mathcal{I}_B$. Additionally, we do not promise items to agents once all items are irrevocably allocated. The price for any feasible trade is calculated in an agent-specific way extending the method of balanced thresholds by \citet{10.1145/2213977.2213991} and balanced prices by \citet{DBLP:conf/focs/DuettingFKL17} to two-sided markets.

	Recall that $A_B$ contains all buyers who receive an item and $A_S$ contains all sellers who irrevocably keep their item. By $r$ we denote the number of items which may  be allocated to buyers in total, i.e. which are not irrevocably kept by a seller, so $r = |S| - |A_S|$. Observe that $r$ is decreasing in our mechanism every time a seller decides to irrevocably keep her item. Given the matroid over the set of buyers, we need to ensure that we do not pick more than $r$ buyers in our mechanism. 

	Fixing a valuation profile $\mathbf{v}$, we let $\OPT_B(\mathbf{v} | A_B, r) \in \arg\max_{B' \subseteq B \setminus A_B, B' \cup A_B \in \mathcal{I}_B, \lvert B' \cup A_B \rvert \leq r} \sum_{i \in B'} v_i$. That is, $\OPT_B(\mathbf{v} | A_B, r)$ denotes the following allocation. Assume that we are only allowed to assign items to buyers (not to sellers) and we have already allocated items to buyers in $A_B$ and at most $r$ items can be allocated to buyers in total. Then $\OPT_B(\mathbf{v} | A_B, r)$ is the allocation that maximizes the welfare increase.
	The value of this partial allocation is denoted by $\mathbf{v} \left( \OPT_B(\mathbf{v} | A_B, r) \right)$. Further, we define $\OPT_B(\mathbf{v}) = \OPT_B(\mathbf{v} | \emptyset, |S|)$ to be the optimal allocation of items to buyers.
	
	The threshold of buyer $i$ is defined with respect to the current state of $A_B$ and the number of items $r$. 
	For a fixed valuation profile $\mathbf{v}$, let \[ p_{i}( A_B, r,\mathbf{v}) = \mathbf{v} \left( \OPT_B(\mathbf{v} | A_B, r) \right) - \mathbf{v} \left( \OPT_B(\mathbf{v} | A_B \cup \{ i \}, r ) \right)  \]
	if $A_B \cup \{ i \} \in \mathcal{I}_B$ and $| A_B \cup \{i\} | \leq r$. So, $p_{i}( A_B, r,\mathbf{v})$ is the difference in welfare which we can achieve by allocating $r$ items to buyers given we have already allocated items to buyers in $A_B$ and $A_B \cup \{i\}$ respectively. To simplify notation, we define $p_{i}( A_B, r,\mathbf{v}) = \infty$ if $A_B \cup \{ i \} \not\in \mathcal{I}_B$ or $| A_B \cup \{i\} | > r$.
	
	Based on this, define buyer $i$'s threshold as \[ p_{i}( A_B, r) = \Ex[\mathbf{\widetilde{v}} \sim \mathcal{D}]{p_{i}( A_B, r, \mathbf{\widetilde{v}})} \enspace. \]
	For a seller $j$, we set the seller-specific threshold to \[ p_j = \Ex[\widetilde{v}_j \sim \mathcal{D}_j]{\widetilde{v}_j} \] which is simply the expected value of the distribution of seller $j$'s value for an item. Now, fix a buyer-seller-pair $(i,j)$ which is available for trading and denote the price for a trade between $i$ and $j$ by \[ p_{i,j}(A_B, r) \coloneqq \frac{1}{3} \left( p_{i}(A_B, r) + p_j \right) \coloneqq \frac{1}{3} \left(  \Ex[\mathbf{\widetilde{v}} \sim \mathcal{D}]{ p_{i}( A_B, r, \mathbf{\widetilde{v}}) } + \Ex[\widetilde{v}_j \sim \mathcal{D}_j]{\widetilde{v}_j} \right) \enspace. \] 
			
\subsection*{Properties of Our Mechanism}

	Note that our mechanism ensures that the final allocation is a feasible solution with respect to the matroid constraint on the buyers' side as we discard any buyer who cannot be added feasibly to the set of allocated agents. Further, we do not allocate more than $|S|$ items in total among all agents in our allocation process.

\begin{theorem} \label{Theorem:Matroid_sbb}
	The mechanism for matroid double auctions is DSBB, DSIC and IR for all buyers and sellers and a $\frac{1}{3}$-approximation to the optimal social welfare.
\end{theorem}

By construction, Mechanism \ref{Matroid_mechanism_sbb} consists of several bilateral trades, where an item is transfered from seller $j$ to buyer $i$ and a price of $p_{i,j}$ is paid by buyer $i$, received by seller $j$, so the mechanism satisfies DSBB. \\ We offer any buyer the possibility to participate in a trade at most once, so DSIC and IR for buyers follows directly. Also IR for sellers is rather simple as we ask seller $j$ every time if she would like to participate in a trade for a given price. In order to show DSIC for sellers, we have to exploit the order in which trades are offered. By this choice, prices offered to a fixed seller are only non-increasing as the mechanism evolves. As a consequence, selling the item as early as possible is only beneficial for a seller (if she would like to sell the item at all). Truthfulness follows as misreporting the value for an item might allow or block unfavorable trades.

\input{matroide_full_information}

%% file: matroide_full_information.tex
In order to illustrate the proof concerning the approximation guarantee, we give a proof for the simplified full information setting. That is, the value $v_i$ of an agent is not a random variable anymore, but rather deterministic. The general case with incomplete information can be found in Appendix \ref{appendix:matroids_sbb}. In the full information setting, the price for a feasible trade between buyer $i$ and seller $j$ simplifies to $\frac{1}{3} \left( \mathbf{v}(\OPT_B(\mathbf{v} | A_{B,ij}, r_{ij} )) - \mathbf{v}(\OPT_B(\mathbf{v} | A_{B,ij} \cup \{i\}, r_{ij} )) + v_j \right)$, where $A_{B,ij}$ and $r_{ij}$ are the states of $A_B$ and $r$ as we consider buyer $i$ and seller $j$ for a trade. First, note that any agent who keeps or purchases an item has a value exceeding some price. So for any agent $i \in A$, there is a price $P_i$ which agent $i$'s value did exceed when we added $i$ to $A$. For sellers to which we did not offer any trade in our mechanism, we set $P_i$ to zero as they keep their items anyway; for buyers who cannot be feasibly added to our set of chosen agents, we set $P_i$ to infinity. We split the social welfare achieved by our mechanism in two parts, calling them \emph{base value} and \emph{surplus}: \[ \sum_{i \in A} v_i = \sum_{i \in A} P_i + \sum_{i \in A} \left( v_i - P_i \right) \] Now, we bound each of these quantities separately. \\ When irrevocably allocating an item during the offer of a trade to buyer-seller-pair $(i,j)$, either the seller keeps the item or the buyer purchases it. In the first case, we reduce $r$ by one, in the second, we add $i$ to $A_B$. In order to bound the loss incurred by a seller keeping her item, observe that 
\begin{align*}
	\mathbf{v}(\OPT_B(\mathbf{v} | A_{B,ij}, r_{ij} )) - \mathbf{v}(\OPT_B(\mathbf{v} | A_{B,ij} \cup \{i\}, r_{ij} )) \\ \geq \mathbf{v}(\OPT_B(\mathbf{v} | A_{B,ij}, r_{ij} )) - \mathbf{v}(\OPT_B(\mathbf{v} | A_{B,ij}, r_{ij}-1 )) \enspace.
\end{align*}

As prices in the next iteration are computed with respect to $r_{ij}-1$ and $A_{B,ij}$, the loss in the buyers' optimal welfare when allocating an item to a seller is bounded by the buyer's contribution to the price. Summing the prices which we offered to agents in $A_B \cup A_S$ combined with this bound leads to a telescopic sum over the buyers' thresholds in the prices. Therefore, we can derive a bound of  \begin{equation}
\label{eq:basevalue} \sum_{i \in A} P_i \geq \frac{1}{3} \mathbf{v}(\OPT_B(\mathbf{v})) - \frac{1}{3} \mathbf{v}(\OPT_B(\mathbf{v} | A_B , r )) + \frac{1}{3} \sum_{j \in S \setminus \Msell} v_j \end{equation} for the base value.

Concerning the surplus, we consider buyers and sellers separately. For the sellers, note that any seller who remains in $\Msell$ after the mechanism keeps her item. Therefore, the contribution to the surplus is $v_j$ for any $j \in \Msell$. In the incomplete information setting, this turns out to be much more involved and a more sophisticated argument needs to be applied. As a consequence, we are only able to bound the sellers' surplus from below via \begin{equation}\label{eq:sellersurplus} \sum_{i \in A_S \cup \Msell} \left( v_i - P_i \right) \geq \frac{2}{3} \sum_{j \in \Msell} v_j - \frac{1}{3} \mathbf{v}(\OPT_B(\mathbf{v} | A_B , r )) \enspace. \end{equation}
For the buyers, we note that the prices for a fixed buyer are only non-decreasing as the allocation process evolves (see Lemma \ref{lemma:increasing_prices_sbb_buyers} which is a generalized version of a lemma in \citet[][Lemma 3]{10.1145/2213977.2213991}). Further, any buyer to which we offer a trade gets an item if her value exceeds her price. Using this, we can bound the surplus of any buyer $i$ to which we proposed a trade via \[ (v_i - P_i)^+ = \left( v_i - p_{i,j_i}(A_{B,ij_i},r_{ij_i})  \right)^+ \geq \left( v_i - p_{i,j_i}(A_{B},r)  \right)^+ \geq  \left( v_i - \min_{j \in \Msell} p_{i,j}(A_{B},r)  \right)^+  \] where we denote by $j_i$ the seller which is matched to buyer $i$. Now, we consider all buyers which are in $\OPT_B(\mathbf{v} | A_B,r)$. Any of these buyers could have purchased an item if her value had exceeded the price. To see this, note that $A_B \cup \OPT_B(\mathbf{v} |A_B,r)$ needs to be independent. Further, if $\OPT_B(\mathbf{v} |A_B,r) \neq \emptyset$, we have that $r > 0$ and so there are still items available after running the mechanism. As a consequence, any agent $i \in \OPT_B(\mathbf{v} |A_B,r)$ has a surplus of $\left( v_i - P_i \right)^+$ which is positive only if $i \in A_B$ after running the mechanism. For a buyer who does not exceed her price, this is zero as is her contribution to the surplus. Summing the surplus of all these buyers implies a lower bound on the overall buyers' surplus of \[ \sum_{i \in A_B } \left(v_i - P_i \right) \geq \sum_{i \in \OPT_B(\mathbf{v} | A_B,r)} (v_i - P_i)^+ \geq \sum_{i \in \OPT_B(\mathbf{v} | A_B,r)} v_i - \sum_{i \in \OPT_B(\mathbf{v} | A_B,r)} \min_{j \in \Msell} p_{i,j}(A_{B},r) \enspace.  \] 
Having a closer look at the sum of prices, we can apply a proposition from \citet[][Proposition 2]{10.1145/2213977.2213991} on the buyers' contribution in order to derive a suitable bound:
\begin{align*}
	\sum_{i \in \OPT_B(\mathbf{v} | A_B,r)} \min_{j \in \Msell} p_{i,j}(A_{B},r) & \leq \frac{1}{3} \left(  \mathbf{v}(\OPT_B(\mathbf{v} | A_B , r )) + |\Msell| \cdot \min_{j\in \Msell} v_j \right) \\ & \leq \frac{1}{3} \left(  \mathbf{v}(\OPT_B(\mathbf{v} | A_B , r )) + \sum_{j\in \Msell} v_j \right) \enspace.
\end{align*}
And so we get
\[
\sum_{i \in A_B } \left(v_i - P_i \right) \geq \frac{2}{3} \mathbf{v}(\OPT_B(\mathbf{v} | A_B , r )) - \frac{1}{3} \sum_{j\in \Msell} v_j \enspace.
\]
Hence, in combination with \eqref{eq:sellersurplus}, we can lower-bound the overall surplus of all agents via
\begin{align}
	\sum_{i \in A} \left( v_i - P_i \right) &\geq  \frac{2}{3} \sum_{j \in \Msell} v_j - \frac{1}{3} \mathbf{v}(\OPT_B(\mathbf{v} | A_B , r )) + \frac{2}{3}  \mathbf{v}(\OPT_B(\mathbf{v} | A_B , r )) - \frac{1}{3} \sum_{j\in \Msell} v_j\notag\\&  = \frac{1}{3} \sum_{j \in \Msell} v_j + \frac{1}{3} \mathbf{v}(\OPT_B(\mathbf{v} | A_B , r )) \enspace. \label{eq:surplus}
\end{align}
Adding base value \eqref{eq:basevalue} and surplus of all buyers and sellers \eqref{eq:surplus} proves the claim as $ \mathbf{v}(\OPT_B(\mathbf{v})) + \sum_{j \in S} v_j \geq  \mathbf{v}(\OPT(\mathbf{v}))$.

%% file: matroids_wbb.tex
\section{Matroid Double Auctions with Weak Budget-Balance and Online Arrival}
\label{section:matroide_wbb}
We consider matroid constraints like in Section \ref{section:matroide_sbb}. So, we have a set of $n$ buyers $B$ and the buyers who receive an item need to form an independent set in $\mathcal{M}_B = \left( B, \mathcal{I}_B \right)$, buyers have unit-demand valuation functions and sellers are unit-supply, each initially equipped with exactly one identical item, hence $k = m$. In contrast to Section \ref{section:matroide_sbb}, our mechanism allows buyers to arrive online with an adversary specifying the order. The adversary may even adapt the choices depending on the set of already considered agents and their valuations. Again, we simplify notation in this section by interpreting $\mathbf{v}$ as the $|B \cup S|$-dimensional vector over the non-negative real numbers in which each entry corresponds to the value of an agent for being allocated an item denoted by $v_i$. Also, we denote the seller as well as the corresponding item by $j$.

\subsection*{The Mechanism}

\begin{algorithm}[h]
\SetAlgoNoLine
\DontPrintSemicolon
%\KwData{Sellers, indexed $1,\dots,m$ and buyers, indexed $m+1, \dots m+n$}
\KwResult{Set $A = A_B \cup A_S$ of agents to get an item with $A_B \subseteq B$, $A_B \in \mathcal{I}_B$, $A_S \subseteq S$ and $|A| = |S|$}
$A \longleftarrow \emptyset$ ; \qquad $A' \longleftarrow \emptyset$ ; \qquad  $M_{\textnormal{SELL}} \longleftarrow \emptyset$; \qquad $T = \left( 0, \dots, 0 \right)$, where $T$ is a vector of zeros of size $|S|$ \\
  \For{$j \in S$}{
	  \If{$v_j \geq p_j$}{ $A \longleftarrow A \cup \{ j \}$; \qquad $A' \longleftarrow A' \cup \{ j \}$ }
	  \If{$v_j < p_j$}{$M_{\textnormal{SELL}} \longleftarrow M_{\textnormal{SELL}} \cup \{ j \}$; \qquad $ T_j \longleftarrow p_j$}
 
    }
  \For{$i \in B$}{
  	\If{$M_{\textnormal{SELL}} \neq \emptyset$}{
  	\text{\textbf{select}} $j \in M_{\textnormal{SELL}}$ \text{\textbf{arbitrarily}} \\
	  	\If{$p_i \geq T_j$}{
	  		\If{$v_i \geq p_i$}{
	  			$A \longleftarrow A \cup \{ i \}$; \qquad $A' \longleftarrow A' \cup \{ i \}$; \qquad $M_{\textnormal{SELL}} \longleftarrow M_{\textnormal{SELL}} \setminus \{ j \}$ \\
	  			\text{buyer $i$ pays $p_i$ to the mechanism and seller $j$ receives $T_j$} \\
	  			}
	  		}
	  	\If{$p_i < T_j$}{
	  		\If{$v_j \geq p_i$}{
	  			$A \longleftarrow A \cup \{ j \}$; \qquad $A' \longleftarrow A' \cup \{ i \}$; \qquad $M_{\textnormal{SELL}} \longleftarrow M_{\textnormal{SELL}} \setminus \{ j \}$ \\
	  			}
	  		\If{$v_j < p_i$}{ $T_j \longleftarrow p_i$ \\
	  			\If{$v_i \geq p_i$}{
	  				$A \longleftarrow A \cup \{ i \}$; \qquad $A' \longleftarrow A' \cup \{ i \}$; \qquad $M_{\textnormal{SELL}} \longleftarrow M_{\textnormal{SELL}} \setminus \{ j \}$ \\
	  				\text{buyer $i$ pays $p_i$ to the mechanism and seller $j$ receives $T_j$} \\
	  				}
	  			}
	  		}
    } }
  \Return $A \cup M_{\textnormal{SELL}}$ 
\caption{Mechanism for Matroid Double Auctions with Online Arrival}
\label{Matroid_mechanism}
\end{algorithm}

Let $A_B$ be the set of buyers who receive an item, hence $A_B \in \mathcal{I}_B$. Further, $A_S$ denotes the set of sellers who decide to keep the item irrevocably. In addition, we define a set $A' = A_B' \cup A_S'$, which is the set of agents that we eventually used to set the prices.
We first go through all the sellers asking whether seller $j$ wants to irrevocably keep her item or try selling it knowing that she will receive at most an amount of $p_j$ (in case we sell the item). Afterwards, we go through all buyers in any order. When buyer $i$ arrives, we match $i$ to an arbitrary seller who still tries selling her item (if available). 

For this buyer-seller pair, we propose the following trade: Buyer $i$ pays the specific price $p_i$ but seller $j$ only receives $\min \left\{ p_i, T_j\right\}$, where $T_j$ is the lowest price that we have ever offered to seller $j$ up to this point. If seller $j$ does not agree, she irrevocably keeps the item from this point onwards; $j$ is added to $A$ but $i$ is added to $A'$. Otherwise, if buyer $i$ does not agree, she is irrevocably discarded. Seller $j$ might get matched again but the price offered to her can only decrease.

Note that this mechanism does not require any specified order in which we process the agents---even further, the matching which we consider for possible trades can be arbitrary, even determined by an adversary. This is a sharp contrast to Section~\ref{section:matroide_sbb}, where we consider buyer-seller pairs in a tailored way.

\subsection*{The Pricing}
\label{Subsection:Matroid_Pricing}

The key to setting the prices is the set $A' = A_B' \cup A_S'$ with $A_B' \subseteq B$ and $A_S' \subseteq S$, which is maintained in addition to the sets $A_B$ and $A_S$. The idea is that for agents in $A'$ the respective agent-specific price can be charged to someone in our mechanism. For a buyer $i \in A'$, this can mean that the buyer herself received an item and paid for it or that the corresponding seller decided to keep the item.
We calculate prices with respect to the set $A'$ in the spirit of the pricing schemes of matroid prophet inequalities by \citet{10.1145/2213977.2213991} and \citet{DBLP:conf/focs/DuettingFKL17}.

In more detail, concerning the buyers, we set prices to infinity if there are no items available anymore or if buyer $i$ cannot be added to $A_B'$, i.e. $p_i(A') = \infty$ if $A_B' \cup \{ i \} \notin \mathcal{I}_B$ or $M_{\textnormal{SELL}} = \emptyset$. In particular, this pricing only affects the buyers and will never occur as long as we go through the sellers (in the first for-loop) in our mechanism. Hence, for sellers, there will always be a finite seller-specific price.

Now, for any agent who can be feasibly added to $A'$ (i.e. all sellers and all buyers in cases different to the ones mentioned above), we compute prices in the following way. Fix a valuation profile $\mathbf{v}$ and denote by $\OPT(\mathbf{v} | X)$ the set of agents who receive an item in an optimal allocation given that we have already irrevocably allocated items to agents in $X$. In contrast to the mechanism in Section \ref{section:matroide_sbb}, the optimum is now computed over all agents, not only over the set of buyers. The value of this partial allocation is denoted by $\mathbf{v} \left( \OPT(\mathbf{v} | X) \right)$, that is, the sum over the value $v_i$ of all agents $i$ who receive an item. Further, define $\OPT(\mathbf{v}) = \OPT(\mathbf{v} | \emptyset)$. \\ Denote by $A_{i}$ and $A_i'$ the state of set $A$ and $A'$ after processing agent $i$ and let \[ p_i(A_{i-1}',\mathbf{v}) =   \mathbf{v} \left( \OPT(\mathbf{v} | A_{i-1}') \right) - \mathbf{v} \left( \OPT(\mathbf{v} | A_{i-1}' \cup \{ i \} ) \right) \enspace.  \] For any seller and all buyers such that $A_{B; i-1}' \cup \{i\} \in \mathcal{I}_B$, as long as there are items remaining, the price for agent $i$ is computed as \[ p_i(A_{i-1}') = \frac{1}{2} \Ex[\mathbf{\widetilde{v}} \sim \mathcal{D}]{ p_i(A_{i-1}',\mathbf{\widetilde{v}}) } \enspace. \]

This way of setting prices also ensures feasibility with respect to the matroid constraint, meaning that $A_B \in \mathcal{I}_B$. The reason is that $A_B' \supseteq A_B$ at all times as every time we add a buyer to $A_B$, the buyer is also added to $A_B'$. We even have $A_B' \in \mathcal{I}_B$ because buyers have infinite prices as soon as they cannot be feasibly added to $A'$.

\subsection*{Properties of Our Mechanism}

\begin{theorem} \label{Theorem:Matroid}
	The mechanism for matroid double auctions is DWBB, DSIC and IR for all buyers and sellers and $\frac{1}{2}$-competitive with respect to the optimal social welfare.
\end{theorem}

The full proof is deferred to Appendix \ref{appendix:matroid_wbb}. To give a sketch, observe that DWBB is obtained via the price comparison of our mechanism: either a trade between some seller $j$ and some buyer $i$ happens at price $p_i$, or buyer $i$ pays $p_i \geq T_j$ to the mechanism whereas seller $j$ only receives $T_j$. The difference $p_i - T_j$ is extracted and never used again.

Satisfying DSIC and IR for buyers can be seen easily as we only offer a trade to any buyer at most once. Also IR for sellers follows naturally. In order to obtain DSIC for sellers, the key observation is that the amount of money which we may pay to seller $j$ is only non-increasing in the allocation process. Hence, as a seller, if you want to sell you item, you want to do so as early as possible.

In order to prove the desired competitive ratio, we split the contribution to the overall welfare into base value and surplus as follows. At the point in time when an agent is added to the set $A$, this agent is offered some price and her value exceeds this price. The part of the agent's value below this price is the base value, the part above the surplus. We bound each quantity separately and consider the sum over all agents afterwards.

The base value is a telescoping sum covering $\frac{1}{2} \Ex[]{ \mathbf{v} \left( \OPT(\mathbf{v}) \right)} - \frac{1}{2} \Ex[]{ \mathbf{v} \left( \OPT(\mathbf{v}| A') \right) }$. Concerning the surplus, we need a few observations: Note that any seller will keep her item when the value exceeds the initial price. Any buyer who is not in $A' \setminus A$ and whose value exceeds her price will purchase an item. Further, agent-specific prices are only non-decreasing: once offered a trade to agent $i$ at price $p_i$ implies that if we were to consider $i$ later in the mechanism, $p_i$ would be at least as high as before. Therefore, we can lower bound the surplus of any agent using an increased price: the price which we could offer to this agent after running the mechanism. By using properties of matroids to find a suitable upper bound on the prices, we bound the overall surplus via $\frac{1}{2} \Ex[]{ \mathbf{v} \left( \OPT(\mathbf{v}| A' \right) }$ from below. Adding the lower bounds for surplus and base value proves the result.

\subsection*{Upper bound on the Approximation Guarantee}
\label{section:upper_bounds}
Concerning upper bounds, we extend the commonly known instance from prophet inequalities\footnote{The instance for prophet inequalities shows that a competitive ratio of $\frac{1}{2}$ is optimal if the order of agents is determined adversarially.} to two-sided markets: We consider a single seller with one item whose value for the item is zero and two buyers where buyer $1$'s value for the item is $v_1 = 1$ and the second buyer has value $v_2$ which is equal to $\frac{1}{\epsilon}$ with probability $\epsilon$ and $0$ else. Applying the results for prophet inequalities in two-sided environments implies the following: If the mechanism for double auctions cannot control the order in which to process the buyers, a competitive ratio of $\frac{1}{2}$ is tight if there are at least two buyers. In particular, any mechanism for double auctions which allows buyers to arrive online with an adversary determining the order cannot achieve more than half of the expected optimal social welfare. By this considerations, also our competitive ratios of Section \ref{section:xos_sbb} are tight. We remark that this bound does only apply if there are at least two buyers, so the case of bilateral trade is excluded. 

%% file: xos_sbb.tex
\section{Combinatorial Double Auctions with Strong Budget-Balance}
\label{section:xos_sbb}
In combinatorial double auctions, we assume that there is a set of $n$ buyers $B$, a set of $k$ sellers $S$ and a set of $m$ possibly heterogeneous items $M$. Our goal is to reallocate the items among the agents by a suitable \emph{sequential posted-prices mechanism}. To this end, we compute a static and anonymous\footnote{Prices which do not depend on the partial allocation are called \emph{static}. Prices are called \emph{anonymous} if they do not depend on the identity of the agent under consideration.} item price $p_j$ for every $j \in M$. Our mechanism is robust concerning the arrival order of buyers: We can assume buyers to arrive online with adversarial order. Even further, our mechanism can also handle an adaptive adversary who can in each round present one of the remaining agents depending on the allocation we have computed so far.

\subsection*{The Mechanism}

Our mechanism is stated in Algorithm \ref{Xos_mechanism_sbb} and works as follows: Given static and anonymous item prices $p_j$, we first ask any seller $l$ which of her items in $I_l$ she would like to keep if we may give her $p_j$ in exchange and which items she would like to try selling for a price of $p_j$. After this, we have a set of available items $\Msell$ which we try to sell to the buyers now. Therefore, we consider buyers sequentially. As a buyer arrives, we ask her which bundle of available items she would like to purchase. We give this bundle to buyer $i$ and buyer $i$ pays the respective prices to any seller from whom she gets an item. After running the mechanism, all items which are unallocated are returned to their corresponding sellers. 

Using static and anonymous item prices allows us to treat sellers as buyers in a one-sided allocation problem. Therefore, first, we go through the sellers asking which subset of items each would like to sell. In a one-sided market, this corresponds to buyers purchasing exactly the bundles of items which any seller would like to keep. Afterwards, we process the buyers one-by-one and ask which of the remaining items each would like to buy. Overall, by the use of static and anonymous item prices, the results concerning the approximation guarantee via prophet inequalities from \citet{DBLP:conf/soda/FeldmanGL15} directly carry over to combinatorial double auctions.

\begin{algorithm}
	\SetAlgoNoLine
	\DontPrintSemicolon
	\KwResult{Allocation $\mathbf{X} = (X_i)_{i \in B \cup S}$ of items to agents such that for any seller $l$ we have $X_l \subseteq I_l$ and $\bigcup_{i \in B \cup S} X_i = \bigcup_{l \in S} I_l = M$ \newline }
	$X_i \longleftarrow \emptyset$ for all $i \in B \cup S$ \qquad $\Msell \longleftarrow \emptyset$ \\
	\For{$l \in S$}{
		Show prices $p_j$ for each item $j \in I_l$ to seller $l$ \\
		Ask seller $l$ which items she wants to keep or try selling \\
		$X_l \longleftarrow \{ j \in I_l : \text{seller } l \text{ wants to keep item } j \}$ \\
		$\Msell \longleftarrow \Msell \cup \{ j \in I_l : \text{seller } l \text{ tries selling item } j \}$
		}	
	\For{$i \in B$}{
		Show prices $p_j$ for each item $j \in \Msell$ to buyer $i$ \\
		Ask buyer $i$ which items she wants to buy \\
		$X_i \longleftarrow \{ j \in \Msell : \text{buyer } i \text{ wants to buy item } j \}$ \\
		$\Msell \longleftarrow \Msell \setminus X_i$ \\
		Buyer $i$ pays $\sum_{j \in X_i} p_j$ \\ 
		Any seller with $j\in I_l$ for some $j\in X_i$ receives $p_j$ and item $j$ is traded to buyer $i$
		}
	\For{$l \in S$}{
		$X_l \longleftarrow X_l \cup \left( \Msell \cap I_l \right)$
		}
	\Return $\mathbf{X}$ 
	\caption{Mechanism for Combinatorial Double Auctions}
	\label{Xos_mechanism_sbb}
\end{algorithm}

\subsection*{The pricing}

We restrict the class of valuation functions for both, buyers and sellers, to valuations which can be represented by fractionally subadditive (also called XOS) functions.\footnote{A set function $a : 2^M \rightarrow \mathbb{R}_{\geq 0}$ is \textit{additive} if and only if there are numbers $c_1,\dots,c_{m} \in \mathbb{R}_{\geq 0}$ such that for any $T \subseteq M$ we have $a(T) = \sum_{j \in T} c_j$. A set function $v: 2^M \rightarrow \mathbb{R}_{\geq 0}$ is \textit{fractionally subadditive} (also called \textit{XOS}) if and only if there are additive set functions $a_1,\dots,a_t$ such that for every $T \subseteq M$ we have $v(T) = \max_{i \leq t } a_i(T)$. Note that the class of XOS valuation functions contains many other classes, such as submodular, gross-substitutes, unit-demand or additive valuation functions.} In order to compute suitable prices, we mimic the pricing approach from \citet{DBLP:conf/soda/FeldmanGL15} and \citet{DBLP:conf/focs/DuettingFKL17} and apply this to two-sided markets. Let $\ALG$ be an algorithm which allocates all items among the agents. We assume that $\ALG$ can only allocate items to sellers which are in their initial bundle. Fix a valuation profile $\mathbf{v}$ and denote by $\mathbf{Y} \coloneqq (Y_i)_{i \in B \cup S}$ the allocation of $\ALG$ under valuation profile $\mathbf{v}$. As any valuation function $v_i$ satisfies the XOS-property, there are additive functions $a_i$ for any $i \in B \cup S$ such that $v_i(T) \geq a_i(T)$ for any $T \subseteq M$ and $v_i(Y_i) = a_i(Y_i)$. For any $j\in Y_i$, denote by $\SW_{j}(\mathbf{v}) \coloneqq a_i(\{j\})$ which you can interpret as the contribution of item $j$ to the overall social welfare given valuation profile $\mathbf{v}$. In other words, for fixed valuation profile $\mathbf{v}$, we consider the allocation $\mathbf{Y}$, the additive set function $a_i$ which represents $v_i(Y_i)$ and evaluate $a_i$ only for a single item $j \in Y_i$. \\ Now, we compute the price for item $j$ as \[ p_j = \frac{1}{2} \Ex[\widetilde{\mathbf{v}} \sim \mathcal{D}]{\SW_j(\widetilde{\mathbf{v}})} \enspace. \] Observe that these prices are static and anonymous item prices for any item $j \in M$. Further, note that for $\ALG$ we have multiple choices: if we do not care about computational issues, we could e.g. use an algorithm $\OPT$ which computes an optimal allocation with respect to social welfare. 

\subsection*{Properties of Our Mechanism}

We consider two different settings for the chosen classes of valuation functions for buyers and sellers. First, we restrict to the case of unit-supply sellers, i.e. each seller bringing one non-identical item to the market. For buyers, we assume XOS-valuation functions. Note that the valuation functions for sellers can also be represented by fractionally subadditive functions and hence, in order to prove the competitive ratio, we can apply Lemma \ref{Lemma:XOS_Approximation_SBB} which allows to state the following theorem. 

\begin{theorem} \label{Theorem:unit_supply_xos_sbb}
	The mechanism for combinatorial double auctions with unit-supply sellers and buyers having XOS-valuation functions is DSBB, DSIC and IR for all buyers and sellers and $\frac{1}{2}$-competitive with respect to the optimal social welfare for any online adversarial arrival order of agents.
\end{theorem}

For our second result, we assume buyers and sellers to have additive valuation functions. Note that this allows sellers to bring more than one item to the market. Since any additive valuation function can trivially be represented by a fractionally subadditive one, we can again apply Lemma \ref{Lemma:XOS_Approximation_SBB} and hence, state the following theorem.

\begin{theorem} \label{Theorem:additive_sbb}
	The mechanism for combinatorial double auctions with buyers and sellers having additive valuation functions is DSBB, DSIC and IR for all buyers and sellers and $\frac{1}{2}$-competitive with respect to the optimal social welfare for any online adversarial arrival order of agents.
\end{theorem}

Concerning the proof of these theorems, observe the following: In the case of unit-supply sellers, seller $l$ is holding one item $j$ initially. Therefore, seller $l$ maximizes utility by keeping the item if $v_l(\{j\}) \geq p_j$ and trying to sell the item else. Hence, we can interpret seller $l$ as a buyer who buys item $j$ if her value exceeds the price. Also, when considering additive seller valuations combined with additive buyer valuations, we can argue in the following way: Any seller has a value $v_l(\{j\})$ for any $j \in I_l$ and hence, we can rewrite the utility of seller $l$ as $\sum_{j \in X_l} v_l(\{j\}) + \sum_{j \in I_l \setminus X_l} p_j$. Since all buyers also have additive valuations, some buyer $i$ will buy an available item $j$ if and only if $v_i(\{j\}) > p_j$. In the case that for all buyers $v_i(\{j\}) < p_j$, the item is returned to the seller anyway. Hence, it is a dominant strategy to try selling all item for which $v_l(\{j\}) \leq p_j$ and keeping the items with $v_l(\{j\}) > p_j$ in order to maximize utility. Therefore, the seller is deciding on her items in the same way as a buyer who is facing the items in set $I_l$. As a consequence, the competitive ratio directly follows by an application of the results from \citet{DBLP:conf/soda/FeldmanGL15}: Interpret the sellers as buyers which are considered first, offer each to keep any item in $X_l$ and sell the remaining afterwards to all buyers via a sequential posted-prices mechanism. Thus, the competitive ratio from \citet{DBLP:conf/soda/FeldmanGL15} directly carries over to our setting. For the sake of completeness, a formal proof of Theorems \ref{Theorem:unit_supply_xos_sbb} and \ref{Theorem:additive_sbb} is given in Appendix \ref{appendix:xos_sbb}.

%% file: knapsack_sbb.tex
\section{Knapsack Double Auctions with Strong Budget-Balance and Online Customized Arrival}
%\section{Knapsack Double Auctions}
\label{section:knapsack_sbb}

In contrast to Sections \ref{section:matroide_sbb} and \ref{section:matroide_wbb} where we considered a Matroid constraint on the set of buyers, we now work in a setting with a Knapsack constraint. That is, each of the $n$ buyers has a weight $w_i \in [0,1]$. The set of buyers $A_B$ who are allocated an item after our mechanism needs to satisfy $\sum_{i \in A_B} w_i \leq 1$. Again, we assume buyers to be unit-demand and sellers to be unit-supply each bringing exactly one identical item to the market, hence $k = m$. Notation is simplified by interpreting $\mathbf{v}$ as the $|B \cup S|$-dimensional vector with non-negative real entries in which each entry $v_i$ corresponds to the value of an agent for being allocated an item. Also, we denote by $j$ the seller as well as the corresponding item.

\subsection*{The Mechanism}

First of all, note that if $k = 1$, i.e. there is only one seller bringing one item to the market, we can simply run our mechanism from Section \ref{section:xos_sbb} in order to get a $1/2$-competitive mechanism which is DSBB, DSIC and IR. Hence, we will restrict to the case of $k \geq 2$ in the following. Further, we start by a restriction to the case of $w_i \leq \frac{1}{2}$ for all buyers $i \in B$. The general case will be discussed at the end of this section.  \\

\begin{algorithm}[h]
\SetAlgoNoLine
\DontPrintSemicolon
%\KwData{Sellers, indexed $1,\dots,m$ and buyers, indexed $m+1, \dots m+n$}
\KwResult{Set $A = A_B \cup A_S$ of agents to get an item with $A_B \subseteq B$, $\sum_{i \in A_B} w_i \leq 1$, $A_S \subseteq S$ and $|A| = |S|$}
$A \longleftarrow \emptyset$ ; \qquad $W \longleftarrow 0$; \qquad $ i \longleftarrow 1$; \qquad $j \longleftarrow 1$ \\
  \While{$i \leq n$ \textnormal{and} $j \leq k$}{
  	\If{$ W + w_i^\ast > 1$}{
  		$i \longleftarrow i+1$ \\
  		} 
  	\If{$ W + w_i^\ast \leq 1$}{
  		\If{$v_j \geq p_i$}{
  			$A \longleftarrow A \cup \{ j \}$; \qquad $W \longleftarrow W + w_i^\ast$; \qquad $j \longleftarrow j+1$\\
  			}
  		\If{$v_j < p_i$}{
  			\If{$v_i \geq p_i$}{
  				$A \longleftarrow A \cup \{ i \}$; \qquad $W \longleftarrow W + w_i^\ast$; \qquad $j \longleftarrow j+1$ \\
  				transfer item from seller $j$ to buyer $i$ for price $p_i$
  				}
  			$i \longleftarrow i+1$\\
  			}
  		}
  	}

  \Return $ A \cup \{ j' \in S : j \leq j' \leq k \}$ 
\caption{Mechanism for Knapsack Double Auctions with Strong Budget-Balance}
\label{mechanism_knapsack_sbb}
\end{algorithm}

We state our mechanism in Algorithm \ref{mechanism_knapsack_sbb} and give a quick description: We sort buyers in a way such that $w_1 \geq w_2 \geq \dots \geq w_n$, compute artificial weights $w_i^\ast$ for any buyer via $w_i^\ast \coloneqq \max\left( w_i; \frac{1}{k} \right)$ and let the buyer-specific price be \[ p_i \coloneqq \frac{2}{7}\cdot w_i^\ast \cdot \Ex[\widetilde{\mathbf{v}}]{\widetilde{\mathbf{v}}\left( \OPT (\widetilde{\mathbf{v}}) \right)} \enspace, \] where $\OPT(\widetilde{\mathbf{v}})$ denotes the optimal allocation of all items among all agents such that the set of selected buyers satisfies the knapsack constraint. We choose $\widetilde{\mathbf{v}}$ to be drawn independently from the same distribution as $\mathbf{v}$. Further, we initialize $W = 0$ which will be our variable controlling feasibility with respect to $w_i^\ast$. In particular, if for some buyer $i$ we have $W + w_i^\ast > 1$, we will not consider buyer $i$ for a trade. In the other case, i.e. that buyer $i$'s artificial weight $w_i^\ast$ can feasibly be added to $W$, we first ask the current seller $j$ if she wants to keep or try selling the item for price $p_i$. If she considers selling, we ask buyer $i$ if she wants to purchase the item.

\subsection*{Feasibility Considerations}

We need to compute a feasible allocation $A$, i.e. the set $A_B = A \cap B$ needs to be feasible with respect to the knapsack constraint. In our mechanism, we instead compute an allocation with respect to the artificial weights $w_i^\ast$. To see that this is also feasible with respect to the initial weights $w_i$, observe that we always ensure $W + w_i^\ast \leq 1$ for any buyer $i$ to which we propose a trade. Since $w_i \leq w_i^\ast$ for any buyer $i$ and we add $w_i^\ast$ to $W$ any time an item is irrevocably allocated, we ensure $\sum_{i \in A_B} w_i \leq \sum_{i \in A_B} w_i^\ast \leq 1$. Further, every time an item is allocated, we add some $w_i^\ast$ to $W$. Since any $w_i^\ast \geq \frac{1}{k}$, we do never allocate more than $k$ items in total.

\subsection*{Properties of Our Mechanism}

\begin{theorem} \label{Theorem:knapsack_sbb}
	The mechanism for knapsack double auctions is DSBB, DSIC and IR for all buyers and sellers and $\frac{1}{7}$-competitive with respect to the optimal social welfare if all buyers' weights are no larger than half of the total capacity.
\end{theorem}

The generalized version without restrictions on the weights can be achieved in the same way with a small loss in the approximation guarantee.

\begin{theorem} \label{Theorem:knapsack_sbb_unrestricted}
	There is a mechanism for knapsack double auctions which is DSBB, DSIC and IR for all buyers and sellers and $\frac{1}{10}$-competitive with respect to the optimal social welfare.
\end{theorem}

We split the proof of Theorem \ref{Theorem:knapsack_sbb} in the two following lemmas.

\begin{lemma} \label{Lemma:Knapsack_DISC_IR_SBB}
	The mechanism for knapsack double auctions where no buyers demands more than half of the total capacity satisfies DSBB. Further, it is DSIC and IR for all buyers and sellers when ordering buyers such that $w_1 \geq w_2 \geq \dots \geq w_n$.
\end{lemma}

\begin{proof}
	Again, the arguments for DSBB, IR and DSIC for buyers follow in similar ways as in the previous sections. Concerning DSIC for sellers, note that we sorted buyers by non-increasing weight. Hence, the price for trades is non-increasing in the ongoing process. As a consequence, as a seller, you would like to sell your item as early as possible (if you want to sell it at all). Therefore, reporting a lower valuation might end in a trade at some price lower than your actual value. On the other hand, reporting a higher valuation may block a possibly beneficial trade. Overall, misreporting does not increase the seller's utility compared to truth-telling. 
\end{proof}

\begin{lemma}\label{Lemma:Knapsack_sbb_approximation}
	The mechanism for knapsack double auctions where no buyers demands more than half of the total capacity is $\frac{1}{7}$-competitive with respect to the optimal social welfare when ordering buyers such that $w_1 \geq w_2 \geq \dots \geq w_n$.
\end{lemma}

\begin{proof}
	The set of agents who receive an item $A$ depends on $\mathbf{v}$, so we denote by $A(\mathbf{v})$ the set $A$ under valuation profile $\mathbf{v}$. Also $W$ depends on $\mathbf{v}$, so in the same way we denote by $W(\mathbf{v})$ the value of $W$ under valuation profile $\mathbf{v}$. We want to compare $\Ex[\mathbf{v}]{\mathbf{v}(A(\mathbf{v}))}$ to $\Ex[\mathbf{v}]{\mathbf{v}(\OPT(\mathbf{v}))}$. To this end, again, we split the welfare of our algorithm into two parts, the base value and the surplus, and bound each quantity separately. The base value is thereby defined as follows: let agent $i$ receive an item in our mechanism, i.e. $i \in A$. Any buyer who gets an item has paid her buyer-specific price for an item. Any seller who decided to keep her item was asked to keep it for some buyer-specific price. The part of any agent's value which is below this price is denoted the base value. The surplus is the part of any agent's value above this threshold if it exists, otherwise it is zero. \\
	
	\textbf{Base Value}: Our base part of the social welfare is defined via the prices. Summing over all agents in $A(\mathbf{v})$, we can compute the following. In particular, we sum over all prices for which either a buyer purchased an item or a seller irrevocably kept it.
	\begin{align*}
	\Ex[\mathbf{v}]{ \sum_{i \in A(\mathbf{v})} p_i   } = \frac{2}{7} \Ex[\widetilde{\mathbf{v}}]{\widetilde{\mathbf{v}}\left( \OPT (\widetilde{\mathbf{v}}) \right)} \cdot \Ex[\mathbf{v}]{W(\mathbf{v})} \geq \frac{2}{7} \Ex[\widetilde{\mathbf{v}}]{\widetilde{\mathbf{v}}\left( \OPT (\widetilde{\mathbf{v}}) \right)} \cdot \frac{1}{2} \Pr[\mathbf{v}]{W(\mathbf{v})\geq \frac{1}{2}} \enspace.
	\end{align*}
	
	\textbf{Surplus}: 
	We consider buyers and sellers separately and combine their respective contributions to the surplus afterwards. \\
	
	\textit{Sellers}: We observe that a seller $j$ might be matched to some buyer $i$ in the mechanism. Denote by $i_j$ the first buyer that seller $j$ is matched to and let $i_j = \perp$ and $w_{i_j}^\ast = 0$ if seller $j$ is never matched to a buyer. Note that this initial matching is independent of seller $j$'s value. Further, prices are only non-increasing in the ongoing process. Thus, we can bound the expected surplus of seller $j$ by the use of $\mathbf{v}$ and $\mathbf{v}'$ being independent and identically distributed combined with linearity of expectation to get
	\begin{align*}
	\Ex[\mathbf{v}]{\text{surplus}_j } & \geq \Ex[\mathbf{v}]{ \left( v_j - p_{i_j} \right)^+ } \geq \Ex[\mathbf{v}, \mathbf{v}']{ \left( v_j - p_{i_j} \right)^+ \cdot \mathds{1}_{W\left((v_j', \mathbf{v}_{-j})\right) \leq \frac{1}{2} } \cdot \mathds{1}_{j \in \OPT\left( (v_j , \mathbf{v}_{-j}') \right)} } \\ & =  \Ex[\mathbf{v}, \mathbf{v}']{ \left( v_j' - p_{i_j} \right)^+ \cdot \mathds{1}_{W\left(\mathbf{v}\right) \leq \frac{1}{2} } \cdot \mathds{1}_{j \in \OPT\left( \mathbf{v}' \right)} } \enspace.
	\end{align*}
	
	\textit{Buyers}: When considering the buyers, we first argue which circumstances need to be fulfilled such that buyer $i$ gets an item in our mechanism. First, buyer $i$'s value needs to exceed her price $p_i$. Second, there needs to be a time step $t$ such that $W_t + w_i^\ast \leq 1$, where $W_t$ denotes the value of $W$ at time step $t$. Third, there needs to exist a seller $j$ such that $v_j \leq p_i$ as otherwise, there will be no item available for buyer $i$. We make use of the following observation: Buyer $i$ is never asked to purchase an item until we either can offer her an item for price $p_i$ or buyer $i$ becomes infeasible with respect to $W$ and $w_i^\ast$. Therefore, everything happening before this event is independent of buyer $i$'s value. Hence, when considering the value of $W$ on a hallucinated valuation profile $v_i'$ drawn independently from the same distribution as $v_i$, we get the following: if $W\left( (v_i', v_{-i}) \right) \leq \frac{1}{2}$, i.e. the value of $W$ on valuation profile $(v_i', v_{-i})$ is at most $\frac{1}{2}$ after running the mechanism, then buyer $i$ could be feasibly added at the end of the mechanism. As $W$ is only non-decreasing, buyer $i$ could have also been feasibly added at time $t$. Using this, we can bound the surplus of buyer $i$ via 
	\begin{align*}
	\text{surplus}_i \geq \left( v_i - p_i \right)^+ \cdot \mathds{1}_{W\left((v_i', \mathbf{v}_{-i})\right) \leq \frac{1}{2} } \geq \left( v_i - p_i \right)^+ \cdot \mathds{1}_{W\left((v_i', \mathbf{v}_{-i})\right) \leq \frac{1}{2} } \cdot \mathds{1}_{i \in \OPT\left( (v_i , \mathbf{v}_{-i}') \right)}  \enspace.
	\end{align*}
	
	Again, using linearity of expectation as well as exploiting that $\mathbf{v}'$ and $\mathbf{v}$ are independent and identically distributed, we get \[ \Ex[\mathbf{v}]{\text{surplus}_i} \geq \Ex[\mathbf{v}, \mathbf{v}']{\left( v_i' - p_i \right)^+ \cdot \mathds{1}_{W\left(\mathbf{v}\right) \leq \frac{1}{2}} \cdot \mathds{1}_{i \in \OPT\left( \mathbf{v}'\right)} } \enspace. \]
	
	\textit{Combination}: Summing over all buyers and sellers, we can combine the two bounds. We denote by $q_{\mathbf{v}} \coloneqq \Pr[\mathbf{v}]{W\left(\mathbf{v}\right) \leq \frac{1}{2}}$.
	\begin{align*}
	\sum_{i \in B \cup S} \Ex[\mathbf{v}]{ \text{surplus}_i } & \geq \sum_{i \in B } \Ex[\mathbf{v}, \mathbf{v}']{\left( v_i' - p_i \right)^+ \cdot \mathds{1}_{W\left(\mathbf{v}\right) \leq \frac{1}{2} } \cdot \mathds{1}_{i \in \OPT\left( \mathbf{v}'\right)} } \\ & + \sum_{j \in S} \Ex[\mathbf{v}, \mathbf{v}']{\left( v_j' - p_{i_j} \right)^+ \cdot \mathds{1}_{W\left(\mathbf{v}\right) \leq \frac{1}{2} } \cdot \mathds{1}_{j \in \OPT\left( \mathbf{v}'\right)} }\\ & = q_{\mathbf{v}} \cdot \left( \Ex[\mathbf{v}']{ \sum_{i \in \OPT\left( \mathbf{v}'\right) \cap B} \left( v_i' - p_i \right)^+   } +  \Ex[\mathbf{v}']{ \sum_{j \in \OPT\left( \mathbf{v}'\right) \cap S} \left( v_j' - p_{i_j} \right)^+   } \right) \\ & \geq q_{\mathbf{v}} \cdot \left( \Ex[\mathbf{v}']{ \mathbf{v}' \left( \OPT(\mathbf{v}') \right)} - \Ex[\mathbf{v}']{ \sum_{i \in \OPT(\mathbf{v}') \cap B} p_i } - \Ex[\mathbf{v}']{ \sum_{j \in \OPT(\mathbf{v}') \cap S} p_{i_j} } \right)  
	\end{align*}	
	
	Again, for the first equality we use that $\mathbf{v}$ and $\mathbf{v}'$ are independent and the respective terms each only depend on one of the two. In order to bound the sums of prices, we note that also $\OPT$ is restricted with a total capacity of one as well as $\OPT$ can also not allocate more than $k$ items, so \[\sum_{i \in \OPT(\mathbf{v}') \cap B } w_i^\ast \leq \sum_{i \in \OPT(\mathbf{v}') \cap B } w_i + \sum_{i \in \OPT(\mathbf{v}') \cap B } \frac{1}{k} \leq 1+1 = 2 \enspace. \] Further, by the definition of $i_j$, we get that  \[ \sum_{j \in \OPT(\mathbf{v}') \cap S} w_{i_j}^\ast \leq \sum_{j \in S} w_{i_j}^\ast \leq 1  \] as we do not tentatively match a buyer and a seller if $W + w_i^\ast$ exceeds one. \\ Therefore, we can bound the overall surplus by \[ \Ex[\mathbf{v}]{\sum_{i \in B \cup S} \text{surplus}_i } \geq  \Pr[\mathbf{v}]{ W\left(\mathbf{v}\right) \leq \frac{1}{2} } \cdot \left( 1 - 3 \cdot \frac{2}{7} \right) \Ex[\widetilde{\mathbf{v}}]{ \widetilde{\mathbf{v}} \left( \OPT(\widetilde{\mathbf{v}}) \right)} \enspace. \]
	
	Summing the base value and the surplus proves our claim as we can exploit that $\mathbf{v}$, $\mathbf{v}'$ and $\widetilde{\mathbf{v}}$ are independent and identically distributed.
\end{proof}

In order to extend this to the general case in Theorem \ref{Theorem:knapsack_sbb_unrestricted} when $w_i \in [0,1]$ instead of $w_i \leq 1/2$, we split the set of buyers in high- and low-weighted ones and run our constructed mechanism on the latter - for the former, we now use the DSBB-mechanism for matroids from Section \ref{section:matroide_sbb}. High-weights buyers are the ones with $w_i > \frac{1}{2}$, low-weighted ones satisfy $w_i \leq \frac{1}{2}$. Observe that in an instance of high-weighted buyers, we can allocate at most one item which corresponds to a $1$-uniform matroid constraint over the set of buyers. Concerning the use of the mechanism for matroid double auctions, note that we do not need to insist on a offline order of buyers now. We can rather fix the arrival sequence of buyers beforehand as we consider the $1$-uniform matroid over all buyers. This implies that all buyers face the same take-it-or-leave-it buyer-respective contribution to the prices and hence allow easier arguments concerning the properties of the mechanism from Section \ref{section:matroide_sbb}. By this construction, we can formulate Theorem \ref{Theorem:knapsack_sbb_unrestricted}.

%% file: knapsack_wbb.tex
\section{Knapsack Double Auctions with Weak Budget-Balance and Online Arrival}
%\section{Knapsack Double Auctions}
\label{section:knapsack_wbb}

Again, we work in a setting with a Knapsack constraint, so each of the $n$ buyers has a weight $w_i \in [0,1]$ with the constraint that $\sum_{i \in A_B} w_i \leq 1$, buyers have unit-demand valuations and sellers are unit-supply each bringing exactly one identical item to the market, hence $k = m$. Our mechanism can handle online adversarial arrival orders of agents - even with an (adaptive) adversary specifying the order. Notation is simplified by interpreting $\mathbf{v}$ as the $|B \cup S|$-dimensional vector with non-negative real entries in which each entry $v_i$ corresponds to the value of an agent for being allocated an item. Also, we overload notation and denote the seller as well as the corresponding item by $j$.

\subsection*{The Mechanism}

As in Section \ref{section:knapsack_sbb}, note that if $k = 1$, i.e. there is only one seller bringing one item to the market, we can simply run our mechanism from Section \ref{section:xos_sbb} in order to get a $1/2$-competitive mechanism which is DSBB, DSIC and IR. Hence, we will only consider the case of $k \geq 2$ in the following. Further, we start by restricting weights to the case of $w_i \leq \frac{1}{2}$ for all buyers $i \in B$. The general case will be discussed in Appendix \ref{appendix:knapsack_wbb}.  \\

\begin{algorithm}[h]
\SetAlgoNoLine
\DontPrintSemicolon
%\KwData{Sellers, indexed $1,\dots,m$ and buyers, indexed $m+1, \dots m+n$}
\KwResult{Set $A = A_B \cup A_S$ of agents to get an item with $A_B \subseteq B$, $\sum_{i \in A_B} w_i \leq 1$, $A_S \subseteq S$ and $|A| = |S|$}
$A \longleftarrow \emptyset$ ; \qquad $M_{\textnormal{SELL}} \longleftarrow \emptyset$ \\
  \For{$j \in S$}{
	  \If{$v_j \geq p_j$}{ $A \longleftarrow A \cup \{ j \}$ }
	  \If{$v_j < p_j$}{$M_{\textnormal{SELL}} \longleftarrow M_{\textnormal{SELL}} \cup \{ j \}$}
    }
  \For{$i \in B$}{
  	\If{$\sum_{i' \in A} w_{i'}^\ast \leq 1 - w_i^\ast$}{
  		\If{$v_i \geq p_i$}{
  			$A \longleftarrow A \cup \{ i \}$; \\
  			\text{pick one $j \in M_{\textnormal{SELL}}$, transfer item from $j $ to $i$, $i$ pays $p_i$ to mechanism, $j$ receives $p_j$}; \\
  			$M_{\textnormal{SELL}} \longleftarrow M_{\textnormal{SELL}} \setminus \{ j \}$
  			}
    } }
  $A \longleftarrow A \cup M_{\textnormal{SELL}}$ 
\caption{Mechanism for Knapsack Double Auctions with Online Adversarial Arrival}
\label{mechanism_knapsack_wbb}
\end{algorithm}

We state our mechanism in Algorithm \ref{mechanism_knapsack_wbb} and give a quick description: We compute artificial weights $w_i^\ast$ for any buyer and seller in the following way: for all buyers, set $w_i^\ast \coloneqq \max\left( w_i; \frac{1}{k} \right)$ and for all sellers, let $w_j^\ast \coloneqq \frac{1}{k}$. For any agent, i.e. any buyer and seller, we set the agent-specific price to be \[ p_i \coloneqq \frac{2}{5}\cdot w_i^\ast \cdot \Ex[\widetilde{\mathbf{v}}]{\widetilde{\mathbf{v}}\left( \OPT (\widetilde{\mathbf{v}}) \right)} \enspace, \] where $\OPT(\widetilde{\mathbf{v}})$ denotes the optimal allocation of all items among all agents such that the set of selected buyers satisfies the knapsack constraint. We choose $\widetilde{\mathbf{v}}$ to be drawn independently from the same distribution as $\mathbf{v}$. Now, we first go through all sellers asking each if she wants to keep or try selling the item if we might pay an amount of $p_j$ to her later-on. Afterwards, we go through all buyers, asking each of them if she wants to purchase an item for price $p_i$ if buyer $i$ can be feasibly added to the chosen set of agents with respect to the artificial weights $w_i^\ast$. 

\subsection*{Feasibility Considerations}

Arguing about the feasibility of our solution, we can proceed similar to Section \ref{section:knapsack_sbb}, as we again compute a feasible allocation with respect to the artificial weights $w_i^\ast$. To see that this is also feasible with respect to the initial weights $w_i$, note that we ensure $1 \geq \sum_{i \in A} w_i^\ast \geq \sum_{i \in A_B} w_i^\ast \geq \sum_{i \in A_B} w_i$ throughout our mechanism. Further, we need to ensure that we do not allocate more than $k$ items in total. This is mirrored by the fact that $w_i^\ast \geq \frac{1}{k}$ for any agent $i$ and, as we do not allocate items if $\sum_{i \in A } w_i^\ast > 1$, we get that $|A| = k \cdot \sum_{i \in A} \frac{1}{k} \leq k \cdot \sum_{i \in A} w_i^\ast \leq k$. 

\subsection*{Properties of Our Mechanism}

\begin{theorem} \label{Theorem:knapsack_wbb}
	The mechanism for knapsack double auctions is DWBB, DSIC and IR for all buyers and sellers and $\frac{1}{5}$-competitive with respect to the optimal social welfare if all buyers' weights are no larger than half of the total capacity.
\end{theorem}

The proof can be found in Appendix \ref{appendix:knapsack_wbb} as well as a proof for the generalized version:

\begin{theorem} \label{Theorem:knapsack_wbb_unrestricted}
	There is a mechanism for knapsack double auctions which is DWBB, DSIC and IR for all buyers and sellers and $\frac{1}{7}$-competitive with respect to the optimal social welfare for any adversarial online arrival order of agents.
\end{theorem}

%% file: conclusion.tex
\section{Conclusion and Open Problems} \label{Section:Conclusion}

As we have shown, understanding one-sided markets properly also allows to design truthful mechanisms in two-sided environments. Somewhat surprisingly, also guarantees on the competitive ratios for mechanisms carry over despite the fact that optimal social welfare cannot be attained in two-sided markets. These results lead to a couple of open questions for future research.

Our first mechanism is strongly budget-balanced, but requires agents to be carefully selected for trades (offline arrival). In contrast, relaxing the budget-balance constraint and allowing weakly budget-balanced mechanisms enables the construction of a mechanism for matroid double auctions which can deal with adversarial online arrival of buyers. It is an interesting goal to combine the best of these two results and design a (posted-prices) mechanism which is strongly budget-balanced but also a $1/2$-approximation and possibly also able to deal with adversarial online arrival of agents. In the online setting, no mechanism will be better than a $1/2$-approximation. How close can one get to this?

Our results concerning combinatorial double auctions hold for unit-supply sellers paired with buyers with fractionally subadditive valuation functions and in settings when all agents have additive valuation functions. Obtaining results with the same approximation guarantee, budget-balance and incentive compatibility constraints for buyers and sellers with XOS valuation functions would be a very desirable result. \citet{10.1145/3381523} already developed a mechanism which is Bayesian incentive-compatible when combining additive sellers and buyers with XOS valuations. But still, the question remains open if there is a mechanism which is DSIC.

One could also try to further improve beyond approximation factors of $1/2$ by considering agents in a suitable order. It is known that for bilateral trade instances we can obtain an approximation guarantee of $1-1/\e$ and even better \citep{DBLP:journals/corr/BlumrosenD16}. As a natural challenge, is a generalization of this guarantee to matroid and combinatorial double auctions possible? To this end, a useful approach could be to extend the posted pricing techniques from \citet{DBLP:conf/soda/EhsaniHKS18} or \citet{DBLP:conf/sigecom/CorreaFHOV17}, which assume that agents arrive in random order, to two-sided markets in the spirit of our extension of prophet inequalities to two-sided markets.

\subsection*{Acknowledgments}

We thank the anonymous reviewers of EC 2021 for helpful comments on improving the presentation of the paper.

%% file: appendix_bilateral_trade.tex
\section{Appendix: Bilateral Trade via balanced prices}
\label{appendix:bilateral_trade}

In order to give a better understanding on how to apply prophet inequality techniques in two-sided markets, we prove the approximation guarantee for bilateral trade instances by the use of balances prices.

In bilateral trades, there is one seller, initially equipped with one item and one buyer. Let $v_s$ denote the seller's value and $v_b$ denote the buyer's value for the item. Both are drawn independently from some (not necessarily identical) probability distributions.

Our mechanism works as follows: Fix price $p = \frac{1}{2} \cdot \Ex[]{\max \{ v_s, v_b \}}$ and trade the item if and only if $v_b \geq p \geq v_s$. We interpret this mechanism as a sequential posted-prices mechanism with price $p$ as follows: First, ask the seller $s$ if she would like to keep or try selling the item for price $p$. Afterwards, buyer $b$ may buy the item for price $p$ if the seller herself wanted to sell the item. Our mechanism is DSIC, IR and SBB by design. Concerning the approximation guarantee, we can state the following proposition.

\begin{proposition} \label{Proposition:bilateral_trade}
	The bilateral trade mechanism with $p = \frac{1}{2} \cdot \Ex[]{\max \{ v_s, v_b \}}$ is a $\frac{1}{2}$-approximation to the optimal social welfare.
\end{proposition}

We give a proof applying the ideas from prophet inequalities.

\begin{proof}
	We distinguish several cases: the mechanism extracts social welfare if the item is either allocated as $v_s > p$ (the seller initially keeps the item), $v_b \geq p \geq v_s$ (a trade occurs), or  $v_b < p$ and $v_s < p$ (both agents' values do not exceed the price, so the item remains at the seller). Observe that the social welfare achieved by the first two cases is clearly a lower bound on the overall social welfare achieved by the mechanism. In other words, we only consider contributions to social welfare if at least one of the agents exceeds price $p$.	
	
	We begin by splitting the social welfare achieved by the mechanism in base value and surplus. \\
	
	For the base value, observe that if there exists $i \in \{ b,s \}$ with $v_i \geq p$, we get a contribution to social welfare of $p$ (actually, we get $p + (v_i -p)$, but the second summand is considered in the surplus). Hence, 
	\begin{align*}
		\Ex[\mathbf{v}]{\text{Base Value}(\mathbf{v})} = \Pr[\mathbf{v}]{\text{There exists } i \in \{ b,s \} \text{ with } v_i \geq p} \cdot p \enspace.
	\end{align*}
	
	For the surplus, we first argue about the contribution of the seller, afterwards about the buyer.
	Note that the seller may keep the item initially if $v_s \geq p$. As a consequence, we can extract a surplus of $v_s - p$ if this is non-negative. In other words, we get $\left(v_s -p\right)^+$ as a surplus from seller $s$.
	Buyer $b$ can buy the item if the seller did initially agree selling, i.e. $v_s < p$ and if her value $v_b \geq p$ exceeds the price. Hence, we get $\left(v_b -p\right)^+ \mathds{1}_{v_s < p}$ as a surplus. Observe that $v_s$ and $v_b$ are independent and hence, taking the expectation over $\mathbf{v}$, we can bound the expected surplus of buyer $b$ via \[  \Ex[\mathbf{v}]{\textnormal{surplus}_b (\mathbf{v}) } \geq \Ex[\mathbf{v}]{ \left(v_b -p\right)^+ \mathds{1}_{v_s < p} } =  \Ex[\mathbf{v}]{ \left(v_b -p\right)^+} \cdot \Pr[]{v_s < p}  \enspace. \]
	
	Now, observe that $1 \geq \Pr[]{v_s < p} \geq \Pr[]{v_s < p, v_b < p}$, which allows to bound the sum of the seller's and buyer's surplus as 
	\begin{align*}
		\Ex[\mathbf{v}]{\text{Surplus}(\mathbf{v})} & = \Ex[\mathbf{v}]{\textnormal{surplus}_s (\mathbf{v}) } + \Ex[\mathbf{v}]{\textnormal{surplus}_b (\mathbf{v}) } \\ & \geq \left( \Ex[\mathbf{v}]{ \left(v_s -p\right)^+} + \Ex[\mathbf{v}]{ \left(v_b -p\right)^+} \right) \cdot \Pr[]{v_s < p, v_b < p} \enspace.
	\end{align*}
	
	Now, use that $\left(v_s -p\right)^+  + \left(v_b -p\right)^+  \geq \max_{i \in \{s,b\}} \left(v_i -p\right)^+ \geq \max \{ v_s,v_b \} -p$. Further, by our choice of $p$, we get that \[ \Ex[\mathbf{v}]{\text{Surplus}(\mathbf{v})} \geq \Ex[\mathbf{v}]{\max \{ v_s,v_b \} -p} \cdot \Pr[]{v_s < p, v_b < p} = p \cdot \Pr[]{v_s < p, v_b < p}  \enspace.  \] 
	
	Combining base value and surplus, we get
	\begin{align*}
		& \Ex[\mathbf{v}]{\text{Base Value}(\mathbf{v})} + \Ex[\mathbf{v}]{\text{Surplus}(\mathbf{v})} \\ & \geq p \cdot \left(  \Pr[\mathbf{v}]{\text{There exists } i \in \{ b,s \} \text{ with } v_i \geq p}  + \Pr[]{v_s < p, v_b < p}  \right) \\ &  = p \cdot 1 \\&= \frac{1}{2} \cdot \Ex[]{\max \{ v_s, v_b \}} \enspace.
	\end{align*}
\end{proof}

%% file: appendix_matroids_sbb_offline.tex
\section{Appendix: Matroid Double Auctions and strong budget-balance}
\label{appendix:matroids_sbb}

In this section, we give a complete proof of Theorem \ref{Theorem:Matroid_sbb}. We split the proof into the two following lemmas.

\begin{lemma} \label{Lemma:Matroid_DISC_IR_SBB}
	Mechanism \ref{Matroid_mechanism_sbb} for matroid double auctions satisfies DSBB. Further, it is DSIC and IR for all buyers and sellers.
\end{lemma}

\begin{proof}
	We first argue that our mechanism is DSBB. Afterwards, concerning DSIC and IR, we consider buyers and sellers separately.
	\begin{itemize}
		\item \textit{DSBB}: By construction, the mechanism consists of bilateral trades where an item is moved from one seller to one buyer and in exchange, money is transfered from this buyer to the corresponding seller. Any time an item is traded between a buyer $i$ and a seller $j$, we ensure that this trade happens for some fixed price $p_{i,j}$. 
		\item \textit{IR - buyers}: Any buyer has the possibility to reject buying an item if her value does not exceed her price. Hence, it is not harmful to participate in the mechanism.
		\item \textit{IR - sellers}: Any seller holding an item is asked if she wants to keep her item if we give her an amount of $p_{i,j}$ for some $i$ in exchange. She could keep her item, so also for sellers, participating is not harmful.
		\item \textit{DSIC - buyers}: Any buyer is asked at most once in our mechanism if she wants to buy an item for some price which only depends on her probability distribution, but not on her private realization. She can either accept the price and buy an item or reject it. In any case, truth-telling is a dominant strategy for any buyer in order to maximize utility.
		\item \textit{DSIC - sellers}: Fix seller $j$. By construction, the prices which we offer to seller $j$ are only non-increasing in the ongoing process. To see this, note that $A_B$ and $r$ do not change as long as we consider seller $j$ for a trade. Therefore, the thresholds $p_i(A_B,r)$ are non-increasing as we ask (possibly) more and more buyers to trade with seller $j$. Hence, as a seller, you want to sell your item as early as possible (if you want to sell it at all). Therefore, reporting a lower valuation might end in a trade at some price lower than your actual value. On the other hand, reporting a higher valuation may block a trade which would be beneficial for the seller. Overall, misreporting does not increase the seller's utility compared to truth-telling. 
	\end{itemize}
\end{proof}

\begin{lemma}\label{Lemma:Matroid_Approximation_SBB}
	Mechanism \ref{Matroid_mechanism_sbb} for matroid double auctions is a $\frac{1}{3}$-approximation of the optimal social welfare.
\end{lemma}

\begin{proof}
	We start by a quick reformulation of the prices. Assume, we introduced a counter $t$ starting at zero which increases by $1$ in every iteration of the while-loop as soon as a buyer or a seller accepts a price. Every time the counter increases, one item is allocated irrevocably: Either the sellers decides to keep the item or a trade occurs and the item is allocated to the current buyer. Denote by $A_{B,t}$ the state of set $A_B$ (similarly with $A_{S,t}$ for $A_S$ etc.) as the counter shows $t$ (i.e. $t$ items are already allocated) and as before, if $A_{B,t} \cup \{ i \} \in \mathcal{I}_B$ and $|A_{B,t} \cup \{i\} | \leq r_t$, let \[ p_{i,j}(A_{B,t}, r_t) = \frac{1}{3} \left(  \Ex[\mathbf{\widetilde{v}} \sim \mathcal{D}]{ p_{i}( A_{B,t}, r_t, \mathbf{\widetilde{v}}) } + \Ex[\widetilde{v}_j \sim \mathcal{D}_j]{\widetilde{v}_j} \right) \] be the price for buyer-seller-pair $(i,j)$. Otherwise, as already mentioned, we will not consider buyer $i$ and set any price $p_{i,j}$ for trades offered to buyer $i$ to infinity. Note that this formulation is equivalent to our initial definition of the prices but rather allows to refer to the $t$-th irrevocably allocated item. \\
	The set of agents who receive an item $A$ depends on $\mathbf{v}$, so we denote by $A(\mathbf{v})$ the set $A$ under valuation profile $\mathbf{v}$ (the same for $A_{B,t}(\mathbf{v})$ and $A_{S,t}(\mathbf{v})$ etc.). We want to compare $\Ex[\mathbf{v}]{\mathbf{v}(A(\mathbf{v}))}$ to $\Ex[\mathbf{v}]{\mathbf{v}(\OPT(\mathbf{v}))}$. To this end, we split the welfare of our algorithm into two parts, the base value and the surplus, and bound each quantity separately. (When thinking about one-sided markets, this corresponds to revenue and utility of buyers.) The base value is thereby defined as follows: let agent $i$ receive an item in our mechanism, i.e. $i \in A$. Any buyer who gets an item has paid some price for the item. Any seller who decided to keep her item was asked to keep it for some specific price. The part of agent $i$'s value which is below this price is denoted the base value. The surplus is the part of agent $i$'s value above this threshold if it exists, otherwise it is zero. There might be sellers who are left unconsidered in our mechanism, i.e. we did never ask them if they would like to participate in a trade. These sellers keep their item without any contribution to the base value in our calculations. All their value is considered in the surplus. \\
	
	\textbf{Base Value}: As said, all buyers and sellers who are irrevocably allocated an item (i.e. which are in $A$ before adding the remaining sellers) have a value which exceeds some price. For any agent $i$, denote this price by $P_i$. Further, every time the counter $t$ increases, we are allocating an item irrevocably in our mechanism. \\
	
	As a first step, we need to argue about the two different scenarios which can occur in our mechanism as an item is allocated after offering a trade to buyer $i$ and seller $j$ with counter $t$. On the one hand, a trade may occur and buyer $i$ is allocated seller $j$'s item. In this scenario, the prices in the next iteration(s) with counter $t+1$ are computed with respect to $A_{B,t+1} = A_{B,t} \cup \{i\}$ and $r_{t+1} = r_{t}$. In addition, seller $j$ is not available for a trade anymore. On the other hand, seller $j$ may keep the item, so we compute prices at counter $t+1$ with respect to $A_{B,t+1} = A_{B,t}$ and $r_{t+1} = r_{t} - 1$. Note that our prices are adapted to mirror the first scenario. Taking the expectation over the inequality from Lemma \ref{lemma:bound_pricing_scenarios_rank}, we see that the impact of the second scenario (i.e. a seller keeping the item) can be bounded by the loss of the first one concerning the optimal social welfare. \\
	
	Fixing a valuation profile $\mathbf{v}$ and summing over all agents in $A(\mathbf{v})$ in the order that they were added to $A$ is equivalent to summing over all steps in which we increased the counter $t$. Denote by $i_t$ and $j_t$ the buyer and seller considered in this particular time step.
	
	This allows to compute the following by a telescopic sum argument:
	
	\begin{align*}
	\sum_{i \in A_B(\mathbf{v}) \cup A_S(\mathbf{v})} P_i & = \sum_{t} \frac{1}{3} \left(  \Ex[\mathbf{\widetilde{v}} \sim \mathcal{D}]{ p_{i_t}( A_{B,t}, r_t, \mathbf{\widetilde{v}}) } + \Ex[\widetilde{v}_{j_t} \sim \mathcal{D}_{j_t}]{\widetilde{v}_{j_t}} \right) \\ & = \frac{1}{3} \sum_{t} \left( \Ex[\mathbf{\widetilde{v}}]{ \mathbf{\widetilde{v}} \left( \OPT_B(\mathbf{\widetilde{v}}| A_{B,t}, r_t)  \right) - \mathbf{\widetilde{v}} \left( \OPT_B(\mathbf{\widetilde{v}} |  A_{B,t} \cup \{ i_t \}, r_t ) \right) } + \Ex[\widetilde{v}_{j_t} \sim \mathcal{D}_{j_t}]{\widetilde{v}_{j_t}} \right) \\ & \stackrel{(\star)}{\geq}  \frac{1}{3} \left( \Ex[\mathbf{\widetilde{v}}]{ \mathbf{\widetilde{v}} \left( \OPT_B(\mathbf{\widetilde{v}} \right)} - \Ex[\mathbf{\widetilde{v}}]{ \mathbf{\widetilde{v}} \left( \OPT_B(\mathbf{\widetilde{v}}| A_B(\mathbf{v}), r(\mathbf{v}) ) \right) } \right) + \frac{1}{3}\sum_{t} \Ex[\widetilde{v}_{j_t} \sim \mathcal{D}_{j_t}]{\widetilde{v}_{j_t}}
	\end{align*}
	
	To see why the last inequality $(\star)$ holds, we use Lemma \ref{lemma:bound_pricing_scenarios_rank}. Consider the step with counter $t$. If buyer $i_t$ gets the item, we argued that $A_{B, t+1} = A_{B,t} \cup \{ i_t\}$ and hence, the sum telescopes. On the other hand, if seller $j_t$ decided to keep the item, we note that by Lemma \ref{lemma:bound_pricing_scenarios_rank}, 
	\begin{align*}
	\Ex[\mathbf{\widetilde{v}}]{ \mathbf{\widetilde{v}} \left( \OPT_B(\mathbf{\widetilde{v}}| A_{B,t}, r_t)  \right) - \mathbf{\widetilde{v}} \left( \OPT_B(\mathbf{\widetilde{v}} |  A_{B,t} \cup \{ i_t \}, r_t ) \right) } \\ \geq \Ex[\mathbf{\widetilde{v}}]{ \mathbf{\widetilde{v}} \left( \OPT_B(\mathbf{\widetilde{v}}| A_{B,t}, r_t)  \right) - \mathbf{\widetilde{v}} \left( \OPT_B(\mathbf{\widetilde{v}} |  A_{B,t}, r_t -1 ) \right) }
	\end{align*}and again, the sum telescopes since the prices in the next step are computed with respect to $r_t -1$. \\
	
	Further, denote by $\Msell (\mathbf{v})$ the set $\Msell$ after running our mechanism with valuation profile $\mathbf{v}$. Note that any seller who is not in $\Msell (\mathbf{v})$ either participated in a trade or irrevocably kept the item during our mechanism. Therefore, \[ \sum_{t} \Ex[\widetilde{v}_{j_t} \sim \mathcal{D}_{j_t}]{\widetilde{v}_{j_t}} = \sum_{j \in S \setminus \Msell (\mathbf{v})} \Ex[\widetilde{v}_{j} \sim \mathcal{D}_{j}]{\widetilde{v}_{j}} \enspace. \] Taking the expectation over all valuation profiles $\mathbf{v}$, exploiting linearity of expectation and using that $\mathbf{\widetilde{v}} \sim \mathcal{D}$, we get:
	
	\begin{align*}
	\Ex[\mathbf{v}]{\sum_{i \in A_B(\mathbf{v}) \cup A_S(\mathbf{v})} P_i } \geq \frac{1}{3} \Ex[\mathbf{v}]{ \mathbf{v} \left( \OPT_B(\mathbf{v} ) \right)} - \frac{1}{3} \Ex[\mathbf{v}, \mathbf{\widetilde{v}}]{ \mathbf{\widetilde{v}} \left( \OPT_B(\mathbf{\widetilde{v}}| A_B(\mathbf{v}), r(\mathbf{v}) ) \right) } + \frac{1}{3} \Ex[\mathbf{v}, \widetilde{\mathbf{v}}]{ \sum_{j \in S \setminus \Msell (\mathbf{v})} \widetilde{v}_{j}}
	\end{align*} 
	
	\textbf{Surplus}: The part of the welfare which is not covered by the base value is captured in the surplus. In order to talk about the surplus of any agent who receives an item, we split the set of agents and examine buyers and sellers separately.\\
		
	\textit{Sellers}: Fix seller $j$. Note that any seller who exceeds a price which we offered keeps her item. By construction of our mechanism, seller $j$ is matched to some buyer(s) in the mechanism and asked if she would like to keep or try selling the item for price $p_{i,j}$. Let $i_j$ denote the first buyer to which $j$ is matched in the mechanism. This matching is independent of seller $j$'s actual valuation since it only depends on the prices for seller $j$ and buyer $i_j$. In the case that $i_j$ does not exist (i.e. seller $j$ was never offered a trade), we can simply set $i_j = \perp$ and $p_{i_j,j} = 0$ and apply the same argument. The last price offered to seller $j$ is $P_j$ (maybe 0 if seller $j$ was never offered a trade) and let the counter show $t$ at this point. \\ Note that the prices which we offered to seller $j$ cannot have increased in the process. Hence, the last price $P_j$ which we offered to seller $j$ is clearly upper bounded by the first price $p_{i_j,j}$ which we offered to $j$. Further, by Lemma \ref{lemma:increasing_prices_sbb_sellers}, the price for the trade between $j$ and $i_j$ is only non-decreasing compared to offering a trade between buyer $i_j$ and seller $j$ later in the process again. Therefore, we can bound the surplus of seller $j$ as follows. 
	
	\begin{align*}
	\left( v_j - P_j \right)^+ & \geq \left( v_j - p_{i_j,j}(A_{B,t}(\mathbf{v}),r_t) \right)^+ \geq \left( v_j -  p_{i_j,j}(A_B((v_j',\mathbf{v}_{-j})),r((v_j',\mathbf{v}_{-j}))) \right)^+ \\ & \geq \left( v_j - p_{i_j,j}(A_B((v_j',\mathbf{v}_{-j})),r((v_j',\mathbf{v}_{-j}))) \right)^+ \cdot \mathds{1}_{j \in \Msell(v_j', \mathbf{v}_{-j})} 
	\end{align*}
	
	Taking expectations on both sides and exploiting that $\mathbf{v}$ and $\mathbf{v}'$ are independent and identically distributed allows the following:
	
	\begin{align*}
	\Ex[\mathbf{v}]{  \left( v_j - P_j \right)^+} & \geq \Ex[\mathbf{v}, \mathbf{v}']{\left( v_j - p_{i_j,j}(A_B((v_j',\mathbf{v}_{-j})),r((v_j',\mathbf{v}_{-j}))) \right)^+ \cdot \mathds{1}_{j \in \Msell(v_j', \mathbf{v}_{-j})} } \\ & = \Ex[\mathbf{v}, \mathbf{v}']{\left( v_j' - p_{i_j,j}(A_B(\mathbf{v}),r(\mathbf{v})) \right)^+ \cdot \mathds{1}_{j \in \Msell(\mathbf{v})} }
	\end{align*}
	
	Next, we can sum over all sellers and use linearity of expectation to obtain the following:
	
	\begin{align*}
	\sum_{j \in S} \Ex[\mathbf{v}]{  \left( v_j - P_j \right)^+} & \geq  \Ex[\mathbf{v}, \mathbf{v}']{ \sum_{j \in S}  \left( v_j' - p_{i_j,j}(A_B(\mathbf{v}),r(\mathbf{v})) \right)^+ \cdot \mathds{1}_{j \in \Msell(\mathbf{v})} } \\ & = \Ex[\mathbf{v}, \mathbf{v}']{ \sum_{j \in \Msell(\mathbf{v}) }  \left( v_j' - p_{i_j,j}(A_B(\mathbf{v}),r(\mathbf{v})) \right)^+ } \\ & \geq \Ex[\mathbf{v}, \mathbf{v}']{ \sum_{j \in \Msell(\mathbf{v}) }  v_j' } -  \Ex[\mathbf{v}]{ \sum_{j \in \Msell(\mathbf{v}) }  p_{i_j,j}(A_B(\mathbf{v}),r(\mathbf{v}))  } 
	\end{align*}
	
	Let us pause for a moment and consider the sum over the prices. First of all, note that by construction of our mechanism, at most one seller $j^\ast \in \Msell(\mathbf{v})$ is offered (maybe multiple times) a trade at all. Therefore, any other seller $j$ satisfies that $i_j = \perp$ and hence for all sellers except $j^\ast$, we can set $p_{i_j,j} = 0$. \\ Having a look at the seller $j^\ast \in \Msell(\mathbf{v})$ who is offered a trade (if $j^\ast$ exists), the price for a trade between $j^\ast$ and $i_{j^\ast}$ was well-defined in the iteration that $j^\ast$ and $i_{j^\ast}$ were considered for a trade. Note that $A_B$ and $r$ did not change after this iteration anymore, so if $i_{j^\ast}$ could be feasibly added to $A_B$ at the step we offered a trade, she also can be feasibly added to $A_B$ after the mechanism. Therefore, combining the price given by \[ p_{i_{j^\ast},j^\ast}(A_B(\mathbf{v}), r(\mathbf{v})) = \frac{1}{3} \left(  \Ex[\mathbf{\widetilde{v}} \sim \mathcal{D}]{ p_{i_{j^\ast}}( A_B(\mathbf{v}), r(\mathbf{v}), \mathbf{\widetilde{v}}) } + \Ex[\widetilde{v}_{j^\ast} \sim \mathcal{D}_{j^\ast}]{\widetilde{v}_{j^\ast}} \right) \] with 
	\begin{align*}
	p_{i_{j^\ast}}( A_B(\mathbf{v}), r(\mathbf{v}), \mathbf{\widetilde{v}}) & = \mathbf{\widetilde{v}} \left( \OPT_B(\mathbf{\widetilde{v}} | A_B(\mathbf{v}), r(\mathbf{v})) \right) - \mathbf{\widetilde{v}} \left( \OPT_B(\mathbf{\widetilde{v}} | A_B(\mathbf{v}) \cup \{ i_{j^\ast} \}, r(\mathbf{v}) ) \right) \\ & \leq \mathbf{\widetilde{v}} \left( \OPT_B(\mathbf{\widetilde{v}} | A_B(\mathbf{v}), r(\mathbf{v})) \right)
	\end{align*}
	allows to bound the sum of prices as follows:
	
	\begin{align*}
	\Ex[\mathbf{v}]{ \sum_{j \in \Msell(\mathbf{v}) }  p_{i_j,j}(A_B(\mathbf{v}),r(\mathbf{v}))  } & \leq \frac{1}{3} \Ex[\mathbf{v}, \widetilde{\mathbf{v}}]{ \mathbf{\widetilde{v}} \left( \OPT_B(\mathbf{\widetilde{v}} | A_B(\mathbf{v}), r(\mathbf{v})) \right) } + \frac{1}{3} \Ex[\mathbf{v}]{  \Ex[\widetilde{v}_{j^\ast} \sim \mathcal{D}_{j^\ast}]{\widetilde{v}_{j^\ast}} }  \\ & \leq \frac{1}{3} \Ex[\mathbf{v}, \widetilde{\mathbf{v}}]{ \mathbf{\widetilde{v}} \left( \OPT_B(\mathbf{\widetilde{v}} | A_B(\mathbf{v}), r(\mathbf{v})) \right) } + \frac{1}{3} \Ex[\mathbf{v}, \widetilde{\mathbf{v}}]{ \sum_{j \in \Msell(\mathbf{v}) }  \widetilde{v}_{j} }
	\end{align*}
	
	Now, we use that $\mathbf{v}$, $\mathbf{v}'$ and $\widetilde{\mathbf{v}}$ are independent and identically distributed. Therefore, we can bound the surplus of all sellers by the following expression:
	
	\begin{equation}
	\begin{aligned} \label{equation:utility_sellers}
	\sum_{j \in S} \Ex[\mathbf{v}]{  \left( v_j - P_j \right)^+} & \geq \Ex[\mathbf{v}, \mathbf{v}']{ \sum_{j \in \Msell(\mathbf{v}) }  v_j' } -  \Ex[\mathbf{v}]{ \sum_{j \in \Msell(\mathbf{v}) }  p_{i_j,j}(A_B(\mathbf{v}),r(\mathbf{v}))  } \\ & \geq \frac{2}{3} \Ex[\mathbf{v}, \mathbf{v}']{ \sum_{j \in \Msell(\mathbf{v}) }  v_j' } - \frac{1}{3} \Ex[\mathbf{v}, \widetilde{\mathbf{v}}]{ \mathbf{\widetilde{v}} \left( \OPT_B(\mathbf{\widetilde{v}} | A_B(\mathbf{v}), r(\mathbf{v})) \right) }
	\end{aligned}
	\end{equation}
	
	\textit{Buyers}: First, observe that initially, all buyers can be feasibly added to $A_B$. During the mechanism, buyers may become infeasible at some point in time. Once a buyer cannot be feasibly added anymore, this buyer will remain infeasible for the remainder of the mechanism. On the other hand, if a buyer can be feasibly added at some point in time, she could also be feasibly added at any time before. During our mechanism, we offer trades to all buyers except of those who did become infeasible on the way. Any of the buyers to which we offer a trade for a finite price gets an item if her value exceeds the offered price. As a consequence, we are allowed to consider $\left( v_i - P_i\right)^+$ as the contribution to the surplus for all buyers. Define $P_i$ for buyer $i$ to be infinity if buyer $i$ was not offered a trade in our mechanism due to the fact that buyer $i$ became infeasible. In the same way, if $p_{i,j}(A_B{,r})$ is not well-defined for a buyer due to the fact that this buyer did become infeasible, we defined the price to be infinity. This directly implies a zero contribution to the surplus, so we do not need to focus on these buyers anymore in our considerations. Otherwise, as before, $P_i$ denotes the price which we offered to buyer $i$. \\ Note that by Lemma \ref{lemma:increasing_prices_sbb_buyers}, the prices which are propose to buyer $i$ are non-decreasing as the allocation process proceeds. As said, any buyer who is offered a trade and exceeds her price gets an item in our mechanism. Note that the price which we offered to buyer $i$ only depends on the sellers and all buyers which did arrive before $i$. In particular, being offered a trade and its price are independent of buyer $i$'s value. As a consequence, for all buyers which are offered trades, we are allowed to calculate the following, where $j_t$ denotes the seller which is matched to buyer $i$ in round $t$, i.e. in the round in which buyer $i$ receives an item (if she does). 
	
	\begin{align*}
	\left( v_i - P_i \right)^+ & = \left( v_i - p_{i,j_t}(A_{B,t}(\mathbf{v}),r_t) \right)^+  \geq \left( v_i - \min_{j \in \Msell((v_i',\mathbf{v}_{-i}))} p_{i,j}(A_B((v_i',\mathbf{v}_{-i})),r((v_i',\mathbf{v}_{-i}))) \right)^+ \\ & \geq \left( v_i - \min_{j \in \Msell((v_i',\mathbf{v}_{-i}))} p_{i,j}(A_B((v_i',\mathbf{v}_{-i})),r((v_i',\mathbf{v}_{-i}))) \right)^+ \cdot \mathds{1}_{i \in \OPT_B\left( (v_i, \mathbf{v}_{-i}') | A_B((v_i',\mathbf{v}_{-i})), r((v_i', \mathbf{v}_{-i})) \right)} 
	\end{align*}
	
	Note that if $ \Msell((v_i',\mathbf{v}_{-i}))$ is empty, then the minimum is taken over the empty set and we do not consider buyer $i$ anymore as in this case buyer $i$ cannot be feasibly added to $A_B$. Taking expectations on both sides and exploiting that $\mathbf{v}$ and $\mathbf{v}'$ are independent and identically distributed allows the following:
	
	\begin{align*}
	\Ex[\mathbf{v}]{  \left( v_i - P_i \right)^+} & \geq \mbox{\rm\bf E}_{\mathbf{v}, \mathbf{v}'}\left[ \left( v_i - \min_{j \in \Msell((v_i',\mathbf{v}_{-i}))} p_{i,j}(A_B((v_i',\mathbf{v}_{-i})),r((v_i',\mathbf{v}_{-i}))) \right)^+ \right. \\ & \textcolor{white}{te} \left. \textcolor{white}{.............} \cdot \mathds{1}_{i \in \OPT_B\left( (v_i, \mathbf{v}_{-i}') | A_B((v_i',\mathbf{v}_{-i})), r((v_i', \mathbf{v}_{-i})) \right)} \right] \\ &  =  \Ex[\mathbf{v}, \mathbf{v}']{\left( v_i' - \min_{j \in \Msell(\mathbf{v})} p_{i,j}(A_B(\mathbf{v}),r(\mathbf{v})) \right)^+ \cdot \mathds{1}_{i \in \OPT_B\left(  \mathbf{v}' | A_B(\mathbf{v}), r(\mathbf{v}) \right)} } 
	\end{align*}
	Now, taking the sum over all buyers, we get
	\begin{align*}
	& \textcolor{white}{.....} \sum_{i \in B} \Ex[\mathbf{v}]{  \left( v_i - P_i \right)^+} \\ &\geq \Ex[\mathbf{v}, \mathbf{v}']{ \sum_{i \in B} \left( v_i' - \min_{j \in \Msell(\mathbf{v})} p_{i,j}(A_B(\mathbf{v}),r(\mathbf{v})) \right)^+ \cdot \mathds{1}_{i \in \OPT_B\left(  \mathbf{v}' | A_B(\mathbf{v}), r(\mathbf{v}) \right)} }  \\ & = \Ex[\mathbf{v}, \mathbf{v}']{ \sum_{i \in \OPT_B\left(  \mathbf{v}' | A_B(\mathbf{v}), r(\mathbf{v}) \right)} \left( v_i' - \min_{j \in \Msell(\mathbf{v})} p_{i,j}(A_B(\mathbf{v}),r(\mathbf{v})) \right)^+ } \\ & \geq \Ex[\mathbf{v}, \mathbf{v}']{ \sum_{i \in \OPT_B\left(  \mathbf{v}' | A_B(\mathbf{v}), r(\mathbf{v}) \right)}  v_i'} - \Ex[\mathbf{v}, \mathbf{v}']{ \sum_{i \in \OPT_B\left(  \mathbf{v}' | A_B(\mathbf{v}), r(\mathbf{v}) \right)}  \min_{j \in \Msell(\mathbf{v})} p_{i,j}(A_B(\mathbf{v}),r(\mathbf{v}))} \\ & = \Ex[\mathbf{v}, \mathbf{v}']{ \mathbf{v}' \left( \OPT_B\left( \mathbf{v}' | A_B(\mathbf{v}), r(\mathbf{v})  \right) \right)} - \Ex[\mathbf{v}, \mathbf{v}']{ \sum_{i \in \OPT_B\left(  \mathbf{v}' | A_B(\mathbf{v}), r(\mathbf{v}) \right)}  \min_{j \in \Msell(\mathbf{v})} p_{i,j}(A_B(\mathbf{v}),r(\mathbf{v}))}
	\end{align*}
	
	Let us take a closer look at the sum over the prices. Using Lemma \ref{lemma:matroide_sbb_prices_upper_bound}, we can upper bound the prices as follows:
	
	\begin{align*}
	& \textcolor{white}{.....} \Ex[\mathbf{v}, \mathbf{v}']{ \sum_{i \in \OPT_B\left(  \mathbf{v}' | A_B(\mathbf{v}), r(\mathbf{v}) \right)}  \min_{j \in \Msell(\mathbf{v})} p_{i,j}(A_B(\mathbf{v}),r(\mathbf{v}))} \\ & =  \frac{1}{3}  \Ex[\mathbf{v}, \mathbf{v}']{ \sum_{i \in \OPT_B\left(  \mathbf{v}' | A_B(\mathbf{v}), r(\mathbf{v}) \right)}  \Ex[\mathbf{\widetilde{v}}]{ p_{i}( A_B(\mathbf{v}), r(\mathbf{v}), \mathbf{\widetilde{v}}) }   } +  \frac{1}{3} \Ex[\mathbf{v}]{\left| \Msell(\mathbf{v}) \right| \cdot \min_{j \in \Msell(\mathbf{v})} \Ex[\widetilde{\mathbf{v}}]{\widetilde{v}_{j} } } \\ & \leq \frac{1}{3} \Ex[\mathbf{v}, \mathbf{v}']{ \mathbf{v}' \left( \OPT_B\left( \mathbf{v}' | A_B(\mathbf{v}), r(\mathbf{v})  \right) \right)} + \frac{1}{3} \Ex[\mathbf{v}, \widetilde{\mathbf{v}}]{\sum_{j \in \Msell(\mathbf{v})} \widetilde{v}_{j}}
	\end{align*}
	
	Overall, the surplus of all buyers can be bounded as follows: 
	
	\begin{align} \label{equation:utility_buyers}
	\sum_{i \in B} \Ex[\mathbf{v}]{  \left( v_i - P_i \right)^+} & \geq \frac{2}{3} \Ex[\mathbf{v}, \mathbf{v}']{ \mathbf{v}' \left( \OPT_B\left( \mathbf{v}' | A_B(\mathbf{v}), r(\mathbf{v})  \right) \right)} - \frac{1}{3} \Ex[\mathbf{v}, \widetilde{\mathbf{v}}]{\sum_{j \in \Msell(\mathbf{v})} \widetilde{v}_{j}}
	\end{align}
	
	\textit{Combination}: Having discussed the surplus of buyers and sellers separately, we combine the two bounds in order to bound the total surplus of our mechanism by summing over all buyers and sellers. Therefore, we sum inequalities (\ref{equation:utility_sellers}) and (\ref{equation:utility_buyers}) and use that $\mathbf{v}$, $\mathbf{v}'$ and $\widetilde{\mathbf{v}}$ are independent and identically distributed.
	
	\begin{align*}
	\Ex[\mathbf{v}]{\sum_{i \in B \cup S} \left( v_i - P_i \right)^+} & \geq \frac{2}{3} \Ex[\mathbf{v}, \mathbf{v}']{ \sum_{j \in \Msell(\mathbf{v}) }  v_j' } - \frac{1}{3} \Ex[\mathbf{v}, \mathbf{v}']{ \mathbf{v}' \left( \OPT_B(\mathbf{v}' | A_B(\mathbf{v}), r(\mathbf{v})) \right) }  \\ & + \frac{2}{3} \Ex[\mathbf{v}, \mathbf{v}']{ \mathbf{v}' \left( \OPT_B\left( \mathbf{v}' | A_B(\mathbf{v}), r(\mathbf{v})  \right) \right)} - \frac{1}{3} \Ex[\mathbf{v}, \mathbf{v}']{\sum_{j \in \Msell(\mathbf{v})} v_{j}'} \\ & = \frac{1}{3} \Ex[\mathbf{v}, \mathbf{v}']{ \mathbf{v}' \left( \OPT_B\left( \mathbf{v}' | A_B(\mathbf{v}), r(\mathbf{v})  \right) \right)} + \frac{1}{3} \Ex[\mathbf{v}, \mathbf{v}']{\sum_{j \in \Msell(\mathbf{v})} v_{j}'}
	\end{align*}
	
	\textbf{Combining Base Value and Surplus}:
	
	Adding base value and surplus and again, using that $\mathbf{v}$, $\mathbf{v}'$ and $\widetilde{\mathbf{v}}$ are independent and identically distributed, we can lower bound the social welfare of our mechanism by
	
	\begin{align*}
	\Ex[\mathbf{v}]{\textnormal{Base Value}} + \Ex[\mathbf{v}]{\textnormal{Surplus}} & \geq  \frac{1}{3} \Ex[\mathbf{v}]{ \mathbf{v} \left( \OPT_B(\mathbf{v} ) \right)} - \frac{1}{3} \Ex[\mathbf{v}, \mathbf{\widetilde{v}}]{ \mathbf{\widetilde{v}} \left( \OPT_B(\mathbf{\widetilde{v}}| A_B(\mathbf{v}), r(\mathbf{v}) ) \right) } \\ & \textcolor{white}{te} + \frac{1}{3} \Ex[\mathbf{v}, \widetilde{\mathbf{v}}]{ \sum_{j \in S \setminus \Msell (\mathbf{v})} \widetilde{v}_{j}} + \frac{1}{3} \Ex[\mathbf{v}, \mathbf{v}']{ \mathbf{v}' \left( \OPT_B\left( \mathbf{v}' | A_B(\mathbf{v}), r(\mathbf{v})  \right) \right)} \\  & \textcolor{white}{te} + \frac{1}{3} \Ex[\mathbf{v}, \mathbf{v}']{\sum_{j \in \Msell(\mathbf{v})} v_{j}'} \\  & = \frac{1}{3} \Ex[\mathbf{v}]{ \mathbf{v} \left( \OPT_B(\mathbf{v} ) \right)} + \frac{1}{3} \Ex[\mathbf{v}]{ \sum_{j \in S } v_{j}}
	\end{align*} 
	We can conclude as $ \Ex[\mathbf{v}]{ \mathbf{v} \left( \OPT_B(\mathbf{v} ) \right)} + \Ex[\mathbf{v}]{ \sum_{j \in S } v_{j}} = \Ex[\mathbf{v}]{ \mathbf{v} \left( \OPT_B(\mathbf{v} ) \right) + \sum_{j \in S } v_{j}}$ and for each valuation profile $\mathbf{v}$, we have that $\mathbf{v} \left( \OPT_B(\mathbf{v} ) \right) + \sum_{j \in S } v_{j}$ is an upper bound on the optimal social welfare which can be achieved by allocating the items among all agents. 
	
\end{proof}

In order to conclude the proof of Theorem \ref{Theorem:Matroid_sbb}, we show the remaining lemmas. First, we aim for a bound of $p_{i}( A_B, r,\mathbf{v}) = \mathbf{v} \left( \OPT_B(\mathbf{v} | A_B, r) \right) - \mathbf{v} \left( \OPT_B(\mathbf{v} | A_B \cup \{ i \}, r ) \right)$ with respect to a change in $r$ instead of adding $i$ to the set $A_B$. 

\begin{lemma}\label{lemma:bound_pricing_scenarios_rank}
	Fix any buyer $i$ and a valuation profile $\mathbf{v}$. Let $A_B$ and $r$ be such that $A_B \cup \{i \} \in \mathcal{I}_B$ and $|A_B \cup \{i \}| \leq r$. Then \[  \mathbf{v} \left( \OPT_B(\mathbf{v} | A_B, r) \right) - \mathbf{v} \left( \OPT_B(\mathbf{v} | A_B \cup \{ i \}, r ) \right) \geq  \mathbf{v} \left( \OPT_B(\mathbf{v} | A_B, r) \right) - \mathbf{v} \left( \OPT_B(\mathbf{v} | A_B , r-1 ) \right) \enspace. \]
\end{lemma}

\begin{proof}
	We argue that $\mathbf{v} \left( \OPT_B(\mathbf{v} | A_B \cup \{ i \}, r ) \right) \leq \mathbf{v} \left( \OPT_B(\mathbf{v} | A_B , r-1 ) \right)  $ which immediately proves the claim. Note that any possible choice of agents for $\OPT_B(\mathbf{v} | A_B \cup \{ i \}, r )$ is also a feasible choice for $\OPT_B(\mathbf{v} | A_B, r-1 )$ and hence, the claim follows.
\end{proof}

Second, we show that prices are only non-decreasing in the ongoing process.

\begin{lemma}\label{lemma:increasing_prices_sbb_sellers}
	Fix buyer $i$ and seller $j$. Let the price for trading between buyer $i$ and $j$ be $p_{i,j}(X, r)$ for some set of remaining sellers $\Msell$. Then we have \[ p_{i,j}(X,r) \leq p_{i,j}(X', r') \] for any superset of allocated agents $X' \supseteq X$ and $r' \leq r$.
\end{lemma}

Before proving the lemma, note that this means that for a fixed buyer-seller-pair $(i,j)$, the prices which we consider in our mechanism are only non-decreasing as the process evolves.

\begin{proof}
	First, if a buyer is infeasible with respect to $X$ and $r$, she also is with respect to $X'$ and $r'$, trivially implying the claim. Also if she could be feasibly added with respect to $X$ and $r$, but not to $X'$ with $r'$, the claim is trivial.
	Therefore, it remains to consider the case where both sides of the inequality are finite. So let us consider $X, X'$ and $r, r'$ such that $i$ can feasibly be added and let $j \in \Msell$. 
	By definition, \[ p_{i,j}(X,r) = \frac{1}{3} \left(  \Ex[\mathbf{\widetilde{v}} \sim \mathcal{D}]{ p_{i}( X, r, \mathbf{\widetilde{v}}) } + \Ex[\widetilde{v}_j \sim \mathcal{D}_j]{\widetilde{v}_j} \right) \enspace. \] Note that $X$, $r$ and $i$ only occur in the first summand whereas $j$ only appears in the second one. Since the second summand is equal for both, $p_{i,j}(X,r)$ and $p_{i,j}(X',r')$, we can reduce the problem to showing that the inequality holds for the first summand. We show the inequality pointwise for any $\mathbf{v}$ and conclude by taking the expectation. Therefore, fix a valuation profile $\mathbf{v}$. We show that
	\begin{align*}
		p_{i}( X, r,\mathbf{v}) & = \mathbf{v} \left( \OPT_B(\mathbf{v} | X, r) \right) - \mathbf{v} \left( \OPT_B(\mathbf{v} | X \cup \{ i \}, r ) \right) \\ &  \stackrel{(1)}{\leq} \mathbf{v} \left( \OPT_B(\mathbf{v} | X, r') \right) - \mathbf{v} \left( \OPT_B(\mathbf{v} | X \cup \{ i \}, r' ) \right) \\ &  \stackrel{(2)}{\leq} \mathbf{v} \left( \OPT_B(\mathbf{v} | X', r') \right) - \mathbf{v} \left( \OPT_B(\mathbf{v} | X' \cup \{ i \}, r' ) \right) =  p_{i}( X', r',\mathbf{v})  \enspace.
	\end{align*}
	
	To show inequality $(1)$, we first use that the basis $\OPT_B(\mathbf{v} | X, r')$ can be chosen to be a subset of $\OPT_B(\mathbf{v} | X, r)$. To see this, denote by $\{ b_1,\dots, b_m \}$ the basis $ \OPT_B(\mathbf{v} | X, r)$ in decreasing order of weights. We show that there is an $m'$ such that $\{ b_1,\dots,b_{m'}\}$ is equal to $\OPT_B(\mathbf{v} | X, r')$, where $m'$ is chosen in a way that $|X \cup \{ b_1,\dots,b_{m'} \}| \leq r'$ and that $X \cup \{ b_1,\dots,b_{m'} \}$ has maximum size with respect to this property (i.e. either we have equality or $r'$ is larger than the cardinality of any independent set - in the latter case, we can just choose a basis without considering $r'$). \\ Assume there is a set $\{ b_1',\dots,b_{m'}' \}$ such that $\sum_{k = 1}^{m'} v_{b_k'} > \sum_{k = 1}^{m'} v_{b_k}$, so $\{ b_1,\dots, b_{m'} \}$ would not be a maximum weight basis with respect to $X$ and $r'$. We know that $\{ b_1',\dots,b_{m'}' \}$ also needs to be independent with respect to $X$ and $r$ and further $m' \leq m$. Therefore, there are $m-m'$ elements in $\{ b_1,\dots,b_{m} \}$ which we can add to $\{ b_1',\dots,b_{m'}' \}$ in order to get a basis in the matroid with respect to $X$ and $r$. Denote these $m-m'$ elements with $b_{\pi_1}, \dots, b_{\pi_{m-m'}}$. Note that $\sum_{k=1}^{m-m'} v_{b_{\pi_k}} \geq \sum_{k=m'+1}^{m} v_{b_k} $. Combining this with the sum from above leads to \[ \sum_{k = 1}^{m'} v_{b_k'} + \sum_{k=1}^{m-m'} v_{b_{\pi_k}} \geq \sum_{k = 1}^{m'} v_{b_k'} + \sum_{k = m'+1}^{m} v_{b_k} > \sum_{k = 1}^{m'} v_{b_k} + \sum_{k = m'+1}^{m} v_{b_k} = \sum_{k = 1}^{m} v_{b_k} \enspace, \] which is a contradiction to the fact that $\{b_1,\dots,b_m\}$ is a maximum weight basis in the matroid given $X$ truncated by $r$. \\
	
	Having this, we can argue about the impact of adding $i$ to $X$ on $\{b_1,\dots,b_{m}\}$ and $\{b_1,\dots,b_{m'}\}$ respectively. Consider two parallel executions of the Greedy algorithm computing $ \OPT_B(\mathbf{v} | X, r)$ and $ \OPT_B(\mathbf{v} | X \cup \{i\}, r)$. The first Greedy will compute $\{ b_1 , \dots, b_m \}$ whereas the second Greedy will choose exactly the same elements except for an element $b_i$ for which $\{b_1,\dots,b_i\} \cup \{i\}$ contains a circuit. Therefore, the difference on the left-hand side of inequality $(1)$ is equal to $v_{b_i}$. \\ Applying the same argument for the difference on the right-hand side of inequality $(1)$, there is an element $b_{i'}$ which is chosen in the first Greedy execution but not in the second one as $\{b_1,\dots,b_{i'}\} \cup \{i\}$ contains a circuit in the matroid contracted with $X$ and truncated with $r'$. Therefore, the difference on the right-hand side is equal to $v_{b_{i'}}$. We argue that $b_{i'}$ cannot be later than $b_i$ in the basis $\{b_1,\dots,b_{m}\}$ which allows us to conclude as elements in $b_1,\dots,b_{m}$ are sorted by weight in decreasing order. \\ We show the claim by contradiction, so assume that $b_{i'}$ is an element after $b_{i}$ and $b_{i'}$ is the first element such that $X \cup \{ b_1,\dots,b_{i'} \} \cup \{i\}$ contains a circuit in the matroid truncated with $r'$. Now, $b_{i'}$ is later than $b_i$, so $X \cup \{ b_1,\dots,b_{i'} \} \cup \{i\}$ is a superset of $Y \coloneqq X \cup \{ b_1,\dots,b_{i} \} \cup \{i\}$. Note that $|Y| \leq |X \cup \{ b_1,\dots,b_{i'} \} \cup \{i\}| \leq r'$. By assumption on $b_i$, $X \cup \{ b_1,\dots,b_{i} \} \cup \{i\}$ contains a circuit in the matroid truncated with $r$ and hence also needs to contain a circuit in the matroid truncated with $r'$, so either $b_{i'}$ is not the first element in $\{ b_1, \dots, b_{m'} \}$ which leads to a circuit with $i$ or $b_{i'}$ is before $b_i$ in the order of the basis. In the first case, apply the same argument again, in the second case, we showed the desired contradiction. Since there are only finitely many elements, the iterative application of the argument will terminate and hence, we proved the first inequality.  \\
	
	To see that inequality $(2)$ holds, we consider the matroid $\mathcal{M}$ truncated to rank $r'$. Denote this matroid by $\mathcal{M}_{r'}$. Expressed differently, this is the intersection of the matroid $\mathcal{M}$ with the $r'$-uniform matroid defined on the same ground set. Using Lemma 3 from \citet{10.1145/2213977.2213991}, the function $f_{r'}(Y) = \mathbf{v}(\OPT_B(\mathbf{v} | Y, r'))$ is submodular in $Y$ where now $\OPT_B(\mathbf{v} | Y, r')$ is a maximum weight basis in the matroid $\mathcal{M}_{r'}$. This implies inequality $(2)$. 	
\end{proof}

Next, we consider a fixed buyer $i$. Note that by the order in which we approach the sellers, pricing buyer $i$ is equivalent to choosing the cheapest current seller out of all available ones and compute the price with respect to the current $A_B$ and $r$. In other words, as buyer $i$ arrives, the price we offer is $\min_{j \in T} p_{i,j}(A_B, r)$, where $T$ denotes the set of available sellers.

\begin{lemma}\label{lemma:increasing_prices_sbb_buyers} \textnormal{(Non-decreasing prices for buyers)}
	Fix any buyer $i$. Then for any $T' \subseteq T \subseteq S $, $A_B' \supseteq A_B$ and $r' \leq r$, we have \[ \min_{j \in T} p_{i,j}(A_B, r)  \leq  \min_{j \in T'} p_{i,j}(A_B', r') \enspace. \] 
\end{lemma}

As a short remark, we never delete agents from the set $A_B$ in our mechanism. Further, the number $r$ never increases and sellers are only removed from $\Msell$ and never added. Therefore, in other words, Lemma \ref{lemma:increasing_prices_sbb_buyers} states that for any fixed buyer $i$, the prices are non-decreasing as the allocation process proceeds.

\begin{proof}
	Observe that the minimum over $T$ contains at least any possible seller $j \in T'$. Hence the minimum on the left is taken over a superset of $T'$. Therefore, the claim follows by applying Lemma \ref{lemma:increasing_prices_sbb_sellers}, i.e. $p_{i,j}(A_B, r)$ is non-decreasing with respect to adding agents to $A_B$ and decreasing the number $r$.
\end{proof}

In order to show that \[ \Ex[\mathbf{v}, \mathbf{v}']{ \sum_{i \in \OPT_B\left(  \mathbf{v}' | A_B(\mathbf{v}), r(\mathbf{v}) \right)}  \Ex[\mathbf{\widetilde{v}}]{ p_{i}( A_B(\mathbf{v}), r(\mathbf{v}), \mathbf{\widetilde{v}}) }   } \leq \Ex[\mathbf{v}, \mathbf{v}']{ \mathbf{v}' \left( \OPT_B\left( \mathbf{v}' | A_B(\mathbf{v}), r(\mathbf{v})  \right) \right)} \] we make use of a proposition from \citet{10.1145/2213977.2213991}. Adapted to our setting, we consider the matroid $\mathcal{M}_r$ which is the matroid over the set of buyers $\mathcal{M}_B$ truncated to rank $r$ (recall the construction by intersecting $\mathcal{M}_B$ with the $r$-uniform matroid over the same ground set which is again a matroid). Denote by $\mathcal{I}_r$ the independent sets in $\mathcal{M}_r$. We apply Proposition 2 from \citet{10.1145/2213977.2213991} to our setting.

\begin{lemma}\label{lemma:matroide_sbb_prices_upper_bound} \citep[adapted version of][Proposition 2]{10.1145/2213977.2213991}
	Fix valuation profile $\widetilde{\mathbf{v}}$ and $r$ and let $A_B \in \mathcal{I}_r$. For any from $A_B$ disjoint set $V \in \mathcal{I}_r$ with $A_B \cup V \in \mathcal{I}_r$, it holds \[ \sum_{i \in V} p_i(A_B, r, \widetilde{\mathbf{v}}) \leq \widetilde{\mathbf{v}}\left( \OPT_B (\widetilde{\mathbf{v}} | A_B, r \right)  \enspace. \] 
\end{lemma}

Setting $V = \OPT_B\left( \mathbf{v}' | A_B (\mathbf{v}), r(\mathbf{v}) \right)$ as well as $A_B = A_B(\mathbf{v})$, we get the desired inequality pointwise for any fixed $\mathbf{v}$ and $\mathbf{v}'$. Hence, we can conclude by taking the expectation on both sides, using linearity and the fact that $\widetilde{\mathbf{v}}$ and $\mathbf{v}'$ are independent and identically distributed.

%% file: appendix_matroids_wbb.tex
\section{Appendix: Matroid Double Auctions with weak budget-balance and online arrival}
\label{appendix:matroid_wbb}

We split the proof of Theorem \ref{Theorem:Matroid} in the two following lemmas. Note that we did not make any assumption on the order in which we process the agents within the set of buyers and sellers. In particular, the order in which we process agents in any of the two classes (i.e. buyers or sellers) could be chosen adversarially.

\begin{lemma} \label{Lemma:Matroid_DISC_IR_WBB}
	Mechanism \ref{Matroid_mechanism} for matroid double auctions satisfies DWBB. Further, it is DSIC and IR for all buyers and sellers for any online adversarial order in which buyers and sellers are processed.
\end{lemma}

\begin{proof}
	It is easy to see that the mechanism is IR for buyers and sellers. Further, it is also DSIC for buyers as any buyer is offered a trade at most once. In addition, for sellers, the maximum amount of money $T_j$ we might give to seller $j$ if we sell her item is non-increasing as the process continues. Hence, the mechanism is also DSIC for sellers. 
	
	Concerning DWBB, observe that any time we trade between a buyer and a seller, we ensure that $p_i \geq T_j$ i.e. the money which is put into the market by buyer $i$ is sufficient to pay $T_j$ to seller $j$. As the difference in money $p_i - T_j$ (possibly 0) is never used in our mechanism again, we ensure DWBB. 
\end{proof}

\begin{lemma}\label{Lemma:Matroid_Approximation}
	Mechanism \ref{Matroid_mechanism} for matroid double auctions is $\frac{1}{2}$-competitive with respect to the optimal social welfare for any online adversarial order of buyers.
\end{lemma}

\begin{proof}
	The set of agents who receive an item $A$ depends on $\mathbf{v}$, so we denote by $A(\mathbf{v})$ and $A'(\mathbf{v})$ the sets $A$ and $A'$ under valuation profile $\mathbf{v}$. We want to compare $\Ex[\mathbf{v}]{\mathbf{v}(A(\mathbf{v}))}$ to $\Ex[\mathbf{v}]{\mathbf{v}(\OPT(\mathbf{v}))}$. To this end, again, we split the welfare of our algorithm into two parts, the base value and the surplus, and bound each quantity separately. The base value is thereby defined as follows: let agent $i$ receive an item in our mechanism, i.e. $i \in A$. Any buyer who gets an item has paid her agent-specific price for an item. Any seller who decided to keep her item was asked to keep it for her agent-specific price or for the buyer-specific price she was matched to. The part of any agent $i$'s value which is below this price is denoted the base value. The surplus is the part of any agent $i$'s value above this threshold if it exists, otherwise it is zero. \\
	
	\textbf{Base Value}: Our base part of the social welfare is defined via the prices. Note that all agents who are irrevocably allocated an item (i.e. which are in $A$ before adding the remaining sellers) have a value which exceeds her agent-specific price, except for sellers $j \in A \setminus A'$ whose value for an item exceeds the price of the corresponding buyer. Denote this final price by $P_i$ for agent $i$. For any seller $j \in A \setminus A'$, the corresponding buyer is stored in $A'\setminus A$, so we can replace the seller and the buyer when summing the base value of all agents who get an item. Fixing a valuation profile $\mathbf{v}$ and summing over all agents in $A(\mathbf{v})$ in the order that they were added to $A$, we can compute the following by a telescopic sum argument:
	\begin{align*}
	\sum_{i \in A(\mathbf{v})} P_i & = \sum_{i \in A'(\mathbf{v})} p_i(A_{i-1}'(\mathbf{v})) = \frac{1}{2} \sum_{i \in A'(\mathbf{v})} \Ex[\mathbf{\widetilde{v}}]{ p_i(A_{i-1}'(\mathbf{v}),\mathbf{\widetilde{v}}) } \\ & = \frac{1}{2} \sum_{i \in A'(\mathbf{v})} \Ex[\mathbf{\widetilde{v}}]{ \mathbf{\widetilde{v}} \left( \OPT(\mathbf{\widetilde{v}}| A_{i-1}'(\mathbf{v})) \right) - \mathbf{\widetilde{v}} \left( \OPT(\mathbf{\widetilde{v}} | A_{i-1}'(\mathbf{v}) \cup \{ i \} ) \right) } \\ & = \frac{1}{2} \left( \Ex[\mathbf{\widetilde{v}}]{ \mathbf{\widetilde{v}} \left( \OPT(\mathbf{\widetilde{v}} \right)} - \Ex[\mathbf{\widetilde{v}}]{ \mathbf{\widetilde{v}} \left( \OPT(\mathbf{\widetilde{v}}| A'(\mathbf{v})) \right) } \right)
	\end{align*}
	Taking the expectation over all valuation profiles $\mathbf{v}$, exploiting linearity of expectation and using that $\mathbf{\widetilde{v}} \sim \mathcal{D}$, we get:
	
	\begin{align*}
	\Ex[\mathbf{v}]{\sum_{i \in A(\mathbf{v})} P_i } = \frac{1}{2} \Ex[\mathbf{v}]{ \mathbf{v} \left( \OPT(\mathbf{v} \right)} - \frac{1}{2} \Ex[\mathbf{v}, \mathbf{\widetilde{v}}]{ \mathbf{\widetilde{v}} \left( \OPT(\mathbf{\widetilde{v}}| A'(\mathbf{v})) \right) }
	\end{align*} 
	
	\textbf{Surplus}: We start with two observations which will be helpful later. \\ 
	First of all, agent-specific prices are non-decreasing: For a fixed agent $i$, it holds that $p_i(A_{i-1}') \leq p_i(A_{i'}')$ for any step $i' > i$, so in particular it holds \[ p_i(A_{i-1}') \leq p_i(A') \enspace. \] To see this we use a reduction of our setting to the one-sided case and apply a lemma from \citet{DBLP:conf/focs/DuettingFKL17} which refers to \citet{10.1145/2213977.2213991}. More on this below. \\
	
	Second of all, we will interrupt for a moment and focus on the set $A'(\mathbf{v}) \setminus A(\mathbf{v})$. This set contains all buyers whose agent-specific price was paid by a seller, i.e. the seller decided to keep the item for price $p_i$. Hence, any buyer $i \in A'(\mathbf{v}) \setminus A(\mathbf{v})$ does not get an item in the end, so their surplus is necessarily zero. Note that by construction, any agent $i \in A'(\mathbf{v}) \setminus A(\mathbf{v})$ is a buyer. We observe that any other agent $i \notin A'(\mathbf{v}) \setminus A(\mathbf{v})$ whose value $v_i$ exceeds her corresponding price gets an item in our mechanism. Additionally, there might be some sellers keeping their items in the end and some sellers keeping items for buyer-specific prices later in the process, but we will not take these contributions to the surplus into account. Overall, any agent $i \notin A'(\mathbf{v}) \setminus A(\mathbf{v})$ had the chance to obtain an item in our process if her value exceeded her price. \\ Next, we want to observe why a buyer $i$ is in $A'(\mathbf{v}) \setminus A(\mathbf{v})$. Having a closer look at our algorithm, we see that buyer $i$ is in $A'(\mathbf{v}) \setminus A(\mathbf{v})$ if and only if her buyer-specific price $p_i$, the value $v_j$ of seller $j$ (the seller who is matched to $i$ once she entered the market) and the lowest price $T_j$ offered to seller $j$ satisfy \[ p_i < T_j \textnormal{\quad and \quad} p_i \leq v_j \enspace. \] In particular, the decision whether buyer $i$ is in $A'(\mathbf{v}) \setminus A(\mathbf{v})$ does not depend on her value $v_i$ at all. So, for any deviation of buyer $i$ to $v_i'$, buyer $i$ would end up in $A'(\mathbf{v}) \setminus A(\mathbf{v})$ in the same cases as she would with valuation $v_i$. Therefore, $i \in A'(\mathbf{v}) \setminus A(\mathbf{v})$ holds if and only if $i \in A'((v_i', \mathbf{v}_{-i})) \setminus A((v_i', \mathbf{v}_{-i}))$. \\
	
	Further, any agent who is already contained in the set $A'(\mathbf{v})$ cannot be added to the set $A'(\mathbf{v})$ afterwards once more. Hence, $i \in A'(\mathbf{v})$ implies $i \notin \OPT(\mathbf{w} | A'(\mathbf{v})) $ for any valuation profile $\mathbf{w}$. Expressed as an implication in the other direction, $i \in \OPT(\mathbf{w} | A'(\mathbf{v}) )$ implies $i \notin A'(\mathbf{v})$ and in particular, $i \in \OPT(\mathbf{w} | A'(\mathbf{v}) )$ implies $i \notin A'(\mathbf{v}) \setminus A(\mathbf{v})$. Combining this with the above observation, any agent $i \in \OPT(\mathbf{w} | A'(\mathbf{v}))$ has a considerable surplus if she exceeds her price. \\
	
	% Note that a seller may also keep her item at the end even if her value does not exceed her specific price. We ignore this value for our calculations.	
	We can now consider the surplus of an agent $i$. Let $\mathbf{v}' \sim \mathcal{D}$ be an independently sampled valuation profile. Now, the price for agent $i$ depends on $A_{i-1}'(\mathbf{v})$. But $A_{i-1}'$ only depends on agents $1,\dots,i-1$, so in particular we could replace $v_i$ by $v_i'$ and use that prices are non-decreasing for any fixed agent. Combining this with the above observations, we can bound the surplus of agent $i$ from below as follows: 
	
	\begin{align*}
	\textnormal{surplus}_i & \geq \left( v_i - p_i(A_{i-1}'(\mathbf{v})) \right)^+  \mathds{1}_{i \notin A'(\mathbf{v}) \setminus A(\mathbf{v})} \\ & = \left( v_i - p_i(A_{i-1}'(\mathbf{v})) \right)^+  \mathds{1}_{i \notin A'((v_i', \mathbf{v}_{-i})) \setminus A((v_i', \mathbf{v}_{-i}))} \\ & \geq \left( v_i - p_i(A'( (v_i', \mathbf{v}_{-i})) \right)^+ \mathds{1}_{i \notin A'((v_i', \mathbf{v}_{-i})) \setminus A((v_i', \mathbf{v}_{-i}))} \\&  \geq \left( v_i - p_i(A'( (v_i', \mathbf{v}_{-i})) \right)^+ \mathds{1}_{i \notin A'((v_i', \mathbf{v}_{-i})) \setminus A((v_i', \mathbf{v}_{-i}))} \mathds{1}_{i \in \OPT\left( (v_i, \mathbf{v}_{-i}') | A' ((v_i', \mathbf{v}_{-i})) \right)} \\ & = \left( v_i - p_i(A'( (v_i', \mathbf{v}_{-i})) \right)^+ \mathds{1}_{i \in \OPT\left( (v_i, \mathbf{v}_{-i}') | A' ((v_i', \mathbf{v}_{-i})) \right)}
	\end{align*}
	
	Taking expectations on both sides and exploiting that $\mathbf{v}$ and $\mathbf{v}'$ are independent and identically distributed leads to the following:
	
	\begin{align*}
	\Ex[\mathbf{v}]{\left( v_i - p_i(A_{i-1}'(\mathbf{v})) \right)^+  \mathds{1}_{i \notin A'(\mathbf{v}) \setminus A(\mathbf{v})}} & \geq \Ex[\mathbf{v}, \mathbf{v}']{\left( v_i - p_i(A'( (v_i', \mathbf{v}_{-i})) \right)^+ \mathds{1}_{i \in \OPT\left( (v_i, \mathbf{v}_{-i}') | A' ((v_i', \mathbf{v}_{-i})) \right)}} \\ & = \Ex[\mathbf{v}, \mathbf{v}']{\left( v_i' - p_i(A'(\mathbf{v})) \right)^+ \mathds{1}_{i \in \OPT\left( \mathbf{v}' | A' (\mathbf{v}) \right)}}
	\end{align*}
	
	Summing over all agents (i.e. buyers and sellers), we can lower bound the overall surplus:
	
	\begin{align*}
	\Ex[\mathbf{v}]{\sum_{i \in A(\mathbf{v})} \left( v_i - P_i \right)^+ } & \geq \Ex[\mathbf{v}]{\sum_{i \in B \cup S} \left( v_i - p_i(A_{i-1}'(\mathbf{v})) \right)^+  \mathds{1}_{i \notin A'(\mathbf{v}) \setminus A(\mathbf{v})}} \\ & \geq  \Ex[\mathbf{v}, \mathbf{v}']{ \sum_{i \in \OPT\left( \mathbf{v}' | A' (\mathbf{v}) \right)} \left( v_i' - p_i(A'(\mathbf{v})) \right)^+} \\ & \geq \Ex[\mathbf{v}, \mathbf{v}']{ \sum_{i \in \OPT\left( \mathbf{v}' | A' (\mathbf{v}) \right)} v_i' } - \Ex[\mathbf{v}, \mathbf{v}']{ \sum_{i \in \OPT\left( \mathbf{v}' | A' (\mathbf{v}) \right)} p_i(A'(\mathbf{v})) } \\ & = \Ex[\mathbf{v}, \mathbf{v}']{ \mathbf{v}' \left( \OPT\left( \mathbf{v}' | A' (\mathbf{v}) \right) \right)} - \Ex[\mathbf{v}, \mathbf{v}']{ \sum_{i \in \OPT\left( \mathbf{v}' | A' (\mathbf{v}) \right)} p_i(A'(\mathbf{v})) } \\ & \geq \frac{1}{2} \Ex[\mathbf{v}, \mathbf{v}']{ \mathbf{v}' \left( \OPT\left( \mathbf{v}' | A' (\mathbf{v}) \right) \right)} 
	\end{align*}
	
	The last inequality follows by bounding $\Ex[\mathbf{v}, \mathbf{v}']{ \sum_{i \in \OPT\left( \mathbf{v}' | A' (\mathbf{v}) \right)} p_i(A'(\mathbf{v})) } \leq \frac{1}{2} \Ex[\mathbf{v}, \mathbf{v}']{ \mathbf{v}' \left( \OPT\left( \mathbf{v}' | A' (\mathbf{v}) \right) \right)}$ which we prove below. \\ 
	
	Summing the base value and the surplus proves our claim as we can exploit that $\mathbf{v}'$ and $\widetilde{\mathbf{v}}$ are independent and identically distributed.
\end{proof}

In order to conclude, we need to prove two remaining facts: first, agent-specific prices are non-decreasing, second, we need to show that \[ \Ex[\mathbf{v}, \mathbf{v}']{ \sum_{i \in \OPT\left( \mathbf{v}' | A' (\mathbf{v}) \right)} p_i(A'(\mathbf{v})) } \leq \frac{1}{2} \Ex[\mathbf{v}, \mathbf{v}']{ \mathbf{v}' \left( \OPT\left( \mathbf{v}' | A' (\mathbf{v}) \right) \right)} \enspace. \]

\subsection*{A different view on our prices}

Our prices ensure that the set of buyers $A_B$ who receive an item in our mechanism is an independent set in the matroid, i.e. $A_B \in \mathcal{I}_B$. Additionally, we ensure that we do not promise items to agents once all items are allocated irrevocably. As we are optimizing over a set of agents which is partially (on the buyers' side) equipped with a matroid constraint, we start by extending this to an equivalent setting with a matroid over the whole set of agents, i.e. the ground set of this extended matroid is $B \cup S$. Afterwards, we show a correspondence of our prices to the ones in \citet{10.1145/2213977.2213991} and \citet{DBLP:conf/focs/DuettingFKL17} respectively. This allows to exploit the properties for the prices in one sided-markets. 

There is the matroid $\mathcal{M}_B = \left( B, \mathcal{I}_B \right)$ over the set of buyers. On the sellers' side we construct an artificial matroid by considering the $|S|$-uniform matroid over the set of sellers, denoted by $\mathcal{M}_S = \left( S, \mathcal{I}_S \right)$. Afterwards, we consider the union of the two matroids $\widehat{\mathcal{M}} = (B \cup S, \mathcal{J})$, where a set $I = I_B \cup I_S$ is now independent, if $I_B \in \mathcal{I}_B$ and $I_S \in \mathcal{I}_S$. In order to mirror the feasibility constraint of having only $|S|$ items, we intersect $\widehat{\mathcal{M}}$ with the $|S|$-uniform matroid over $B \cup S$ and denote this matroid by $\mathcal{M}$. Observe that by construction, $\mathcal{M}$ is again a matroid. As a consequence, we can relate all feasible allocations with respect to $\mathcal{M}_B$ to independent sets in the extended matroid $\mathcal{M}$.

Concerning our pricing scheme, first, observe that we calculated prices with respect to the set $A'$ by setting $p_i = \infty$ if $A_B' \cup \{ i \} \notin \mathcal{I}_B$ or if all items are irrevocably allocated. This corresponds to sets which are not independent in the extended matroid $\mathcal{M}$ over the ground set $B \cup S$. A finite price for $i$ (in case $i$ can feasibly be added to $A'$) can also be interpreted in the extended matroid $\mathcal{M}$: If $A' \cup \{ i \} \in \mathcal{I}$, the price for agent $i$ is computed to be $p_i(A_{i-1}')$. Hence, we ensure that $\OPT(\mathbf{v}) \in \mathcal{I}$ and also $X \cup \OPT(\mathbf{v} | X) \in \mathcal{I}$ for any $X \in \mathcal{I}$. Note that in particular, the matroid $\mathcal{M}$ combines the feasibility constraints for buyers and the constraint of having $|S|$ items. As a consequence, computing prices with respect to $\mathcal{M}$ is equivalent to our pricing strategy from Section \ref{Subsection:Matroid_Pricing}.

\subsubsection*{Properties of Prices}

Looking at our optimization problem and in particular on the prices from the viewpoint of the matroid $\mathcal{M}$, we can use Lemma E.2 in \citet{DBLP:conf/focs/DuettingFKL17} to show that agent-specific prices are non-decreasing for any fixed agent $i$.

\begin{lemma} \citep[Lemma E.2]{DBLP:conf/focs/DuettingFKL17}
	Consider any independent sets $X, Y \in \mathcal{I}$ with $X \subseteq Y$. Then, for any agent $i$, we have \[ p_i(X) \leq p_i(Y)  \enspace.  \] 
\end{lemma}

Having that agent-specific prices are only non-decreasing, it remains to show that \[ \Ex[\mathbf{v}, \mathbf{v}']{ \sum_{i \in \OPT\left( \mathbf{v}' | A' (\mathbf{v}) \right)} p_i(A'(\mathbf{v})) } \leq \frac{1}{2} \Ex[\mathbf{v}, \mathbf{v}']{ \mathbf{v}' \left( \OPT\left( \mathbf{v}' | A' (\mathbf{v}) \right) \right)} \] in order to conclude. Again, we use the matroid $\mathcal{M}$ as constructed above and apply a proposition from \citet{10.1145/2213977.2213991}.

\begin{lemma} \citep[Proposition 2]{10.1145/2213977.2213991}
	Fix valuation profile $\widetilde{\mathbf{v}}$ and let $A' \in \mathcal{I}$. For any disjoint set $V \in \mathcal{I}$ with $A' \cup V \in \mathcal{I}$, it holds \[ \sum_{i \in V} p_i(A', \widetilde{\mathbf{v}}) \leq \widetilde{\mathbf{v}}\left( \OPT (\widetilde{\mathbf{v}} | A' \right)  \enspace. \] 
\end{lemma}

Setting $V = \OPT\left( \mathbf{v}' | A' (\mathbf{v}) \right)$ as well as $A' = A'(\mathbf{v})$, we get the desired inequality pointwise for any fixed $\mathbf{v}$ and $\mathbf{v}'$. Hence, we can conclude by taking the expectation and using that $\widetilde{\mathbf{v}}$ and $\mathbf{v}'$ are independent and identically distributed.

%% file: appendix_xos_sbb.tex
\section{Appendix: Combinatorial Double Auctions with strong budget-balance}
\label{appendix:xos_sbb}

In order to prove Theorems \ref{Theorem:unit_supply_xos_sbb} and \ref{Theorem:additive_sbb}, we provide the following lemmas. We start with the DSBB, DSIC and IR properties of our mechanism and conclude by proving the competitive ratio.

\begin{lemma} \label{Lemma:unit_supply_xos_DISC_IR_SBB}
	The mechanism for combinatorial double auctions with unit-supply sellers and buyers having XOS-valuation functions is DSBB, DSIC and IR for all buyers and sellers. The agents' arrival order can be chosen adversarially. 
\end{lemma}

\begin{proof}
	By construction, the mechanism consists of bilateral trades where an item is moved from one seller to one buyer and in exchange, a static and anonymous item price is transfered from this buyer to the corresponding seller. Hence, we satisfy DSBB. In addition, IR is also satisfied as any agent can withdraw. The mechanism is further DSIC for buyers as any buyer is asked once in our mechanism which bundle she wants to purchase. As seller $l$ only has one item and we offer her a price of $p_j$ for the item, the mechanism is also DSIC for sellers.
\end{proof}

\begin{lemma} \label{Lemma:gs_DISC_IR_SBB}
	The mechanism for combinatorial double auctions with buyers and sellers having additive valuation functions is DSBB, DSIC and IR for all buyers and sellers, where the agents' arrival order can be chosen adversarially. 
\end{lemma}

\begin{proof}
	Concerning DSBB, IR and DSIC for buyers, we can copy the arguments from Lemma \ref{Lemma:unit_supply_xos_DISC_IR_SBB}. Also, the mechanism is DSIC for sellers: By additivity, any seller has a value $v_l(\{j\})$ for any $j \in I_l$ and hence, we can rewrite the utility as $\sum_{j \in X_l} v_l(\{j\}) + \sum_{j \in I_l \setminus X_l} p_j$. Since all buyers also have additive valuations, some buyer $i$ will buy an available item $j$ if and only if $v_i(\{j\}) > p_j$. In the case that for all buyers $v_i(\{j\}) < p_j$, the item is returned to the seller anyway. Hence, it is a dominant strategy to try selling all item for which $v_l(\{j\}) \leq p_j$ and keeping the items with $v_l(\{j\}) > p_j$ in order to maximize utility. 
\end{proof}

In order to conclude, we show a lemma bounding the competitive ratio of our mechanism with respect to the social welfare of the algorithm $\ALG$. As said, $\ALG$ can either be an optimal mechanism, leading to the desired $\frac{1}{2}$-competitive mechanism with respect to the optimal welfare, or $\ALG$ can be chosen to be any other approximation algorithm for the optimal social welfare which allocates all items. In the latter case, an $\alpha$-approximation algorithm $\ALG$ leads to an $\frac{\alpha}{2}$-competitive mechanism.

\begin{lemma}\label{Lemma:XOS_Approximation_SBB}
	The mechanism for combinatorial double auctions is $\frac{1}{2}$-competitive with respect to the social welfare of $\ALG$ for any (possibly adversarial) arrival order of buyers and sellers when buyers have valuation functions which can be represented by fractionally subadditive functions and sellers are unit-supply or both have additive valuation functions over item bundles. 
\end{lemma}

\begin{proof}
	In order to show the desired competitive ratio, we mimic the techniques from \citet{DBLP:conf/soda/FeldmanGL15} and \citet{DBLP:conf/focs/DuettingFKL17}. Therefore, we split the contribution to social welfare into base value and surplus and bound each quantity separately. \\ Before we start, note that our mechanism consists of three phases. First, we ask all sellers which items should be sold and which they would like to keep. Afterwards, in the second phase, we ask all buyers which items they would like to buy. In the last phase, all unsold items are returned to the corresponding sellers. Note that the last phase only increases the welfare of our mechanism compared to a mechanism which would stop after the second phase and dispose all unallocated items. We do not consider the increase in welfare in the third phase and argue about the welfare which we have already achieved after the second phase. This is a lower bound on the overall social welfare of our mechanism. \\
	
	\textbf{Base Value}: Note that any item which is irrevocably allocated in the first two phases of the mechanism is allocated to an agent who has an item-specific value at least as high as the item price. Denote by $\A$ the set of irrevocably allocated items after the second phase, i.e. all items which are allocated before the last for-loop in the mechanism where we return unallocated items to their corresponding sellers. Note that $\A$ depends on the valuation profile $\mathbf{v}$. We write $\A(\mathbf{v})$ in order to specify this dependence here. Therefore, we can state the base value as
	\begin{align*}
	\Ex[\mathbf{v}]{\text{Base Value}(\mathbf{v})} = \sum_{j \in M} \Pr[\mathbf{v}]{j \in \A(\mathbf{v})} \cdot p_j \enspace.
	\end{align*}
	
	\textbf{Surplus}: For the surplus, we split the set of agents in buyers and sellers and consider them separately. Note that the sets $(Y_i)_{i \in B \cup S}$ depend on the valuation profile $\mathbf{v}$. Hence, we write $Y_i(\mathbf{v})$ for the bundle of items which are allocated to agent $i$ under valuation profile $\mathbf{v}$. Note that $\ALG$ can only allocate items to seller $l$ which are in $I_l$, i.e. only items which seller $l$ holds at the beginning of the mechanism.  \\
	
	\textit{Sellers}: Fix seller $l$. If $l$ is holding one item $j$ initially, then seller $l$ irrevocably keeps the item if $v_l(\{j\}) \geq p_j$. Therefore, seller $l$ has a considerable surplus if $\left(v_l(\{j\}) - p_j \right)^+ \geq 0$. The same argument extends to the case of additive valuation functions, as seller $l$ will initially keep all items for which $v_l(\{j\}) \geq p_j$ in order to maximize utility (see Lemma \ref{Lemma:gs_DISC_IR_SBB}). Counting the surplus of seller $l$ only for items in $Y_l\left((v_l,\mathbf{v}_{-l}')\right)$ is a feasible lower bound for the surplus of seller $l$ in our mechanism. Here, $\mathbf{v}' \sim \mathcal{D}$ denotes an independent sample.
	Hence, for seller $l$, we can bound the surplus via
	\begin{align*}
	\Ex[\mathbf{v}]{\textnormal{surplus}_l (\mathbf{v}) } &  \geq \Ex[\mathbf{v}, \mathbf{v}']{ \sum_{j \in Y_l\left((v_l,\mathbf{v}_{-l}')\right)} \left( \SW_j\left((v_l,\mathbf{v}_{-l}')\right) -  p_j \right)^+ } = \Ex[\mathbf{v}']{ \sum_{j \in Y_l\left(\mathbf{v}'\right)} \left( \SW_j\left(\mathbf{v}'\right) -  p_j \right)^+ } \enspace.
	\end{align*}
	
	We used that $\mathbf{v}$ and $\mathbf{v}'$ are independent and identically distributed. Additionally, we are able to rewrite $v_l\left( Y_l\left(\mathbf{v}\right) \right)$ as described in Section \ref{section:xos_sbb} via the additive set function $a_l$ and the contribution to social welfare $\SW_j\left(\mathbf{v}\right)$. \\
	
	\textit{Buyers}: Fix buyer $i$. Extending the notation from above, denote by $\A_i(\mathbf{v})$ the set of irrevocably allocated items as agent $i$ is considered in the mechanism. Note that the set $\A_i$ does not depend on $v_i$ but only on the agents which were considered before $i$. Hence, $\A_i(\mathbf{v}) = \A_i\left( (v_i', \mathbf{v}_{-i}) \right)$ for any other valuation $v_i'$ of buyer $i$. Buyer $i$ could purchase the set $Y_i\left((v_i,\mathbf{v}_{-i}')\right) \setminus \A_i\left( (v_i', \mathbf{v}_{-i}) \right)$. As buyer $i$ maximizes utility, the utility which buyer $i$ obtains must be at least as high as the utility when purchasing $Y_i\left((v_i,\mathbf{v}_{-i}')\right) \setminus \A_i\left( (v_i', \mathbf{v}_{-i}) \right)$. As the utility of buyer $i$ is captured in the surplus, we can bound the surplus of buyer $i$ as follows:
	
	\begin{align*}
	\Ex[\mathbf{v}]{\textnormal{surplus}_i (\mathbf{v}) } & \geq \Ex[\mathbf{v}, \mathbf{v}']{ \sum_{j \in Y_i\left((v_i, \mathbf{v}_{-i}')\right) \setminus \A_i\left((v_i', \mathbf{v}_{-i})\right) } \left( \SW_j\left((v_i, \mathbf{v}_{-i}')\right) - p_j \right)^+ } \\ & = \Ex[\mathbf{v}, \mathbf{v}']{ \sum_{j \in Y_i\left(\mathbf{v}'\right) \setminus \A_i\left( \mathbf{v}\right) } \left( \SW_j(\mathbf{v}') - p_j \right)^+ }
	\end{align*}
	
	\textit{Combination}: Next, we sum over all buyers and sellers. Further, we use that once an item is irrevocably allocated, it remains so until the end of the mechanism, hence $\A_i(\mathbf{v}) \subseteq \A(\mathbf{v})$ for any agent $i$ and any valuation profile $\mathbf{v}$. In order to simplify notation, note that $\A_l(\mathbf{v}) \cap I_l = \emptyset$ as we ask seller $l$ which items she wants to keep or try selling. As $Y_l \subseteq I_l$, we know that $Y_l \setminus \A_l(\mathbf{v}) = Y_l$ as seller $l$ arrives.
	
	\begin{align*}
	\Ex[\mathbf{v}]{\sum_{i \in B \cup S} \textnormal{surplus}_i (\mathbf{v})  } & \geq \Ex[\mathbf{v}, \mathbf{v}']{ \sum_{i \in B \cup S} \sum_{j \in M}  \left( \SW_j(\mathbf{v}') - p_j \right)^+ \cdot \mathds{1}_{j \in Y_i\left(\mathbf{v}'\right) } \cdot \mathds{1}_{j \notin \A_i\left( \mathbf{v}\right)} } \\ & \geq \sum_{j \in M} \sum_{i \in B \cup S} \Ex[\mathbf{v}, \mathbf{v}']{ \left( \SW_j(\mathbf{v}') - p_j \right)^+ \cdot \mathds{1}_{j \in Y_i\left(\mathbf{v}'\right) } \cdot \mathds{1}_{j \notin \A\left( \mathbf{v}\right)} } \\ & \geq \sum_{j \in M} \Pr[\mathbf{v}]{j \notin \A\left( \mathbf{v}\right)} \cdot \Ex[\mathbf{v}']{ \sum_{i \in B \cup S}  \left( \SW_j(\mathbf{v}') - p_j \right) \cdot \mathds{1}_{j \in Y_i\left(\mathbf{v}'\right) } } \\ & = \sum_{j \in M} \Pr[\mathbf{v}]{j \notin \A\left( \mathbf{v}\right)} \cdot p_j
	\end{align*}
	
	\textbf{Combining Base Value and Surplus}: Adding base value and surplus together, we get the desired bound:
	
	\begin{align*}
	\Ex[\mathbf{v}]{\text{Base Value}(\mathbf{v})} + \Ex[\mathbf{v}]{\sum_{i \in B \cup S} \textnormal{surplus}_i (\mathbf{v})  } & \geq \sum_{j \in M} \left( \Pr[\mathbf{v}]{j \in \A\left( \mathbf{v}\right)}  + \Pr[\mathbf{v}]{j \notin \A\left( \mathbf{v}\right)}  \right) \cdot p_j \\ & = \sum_{j \in M} p_j = \frac{1}{2} \Ex[\mathbf{v}]{\sum_{i \in B \cup S} v_i(Y_i)}
	\end{align*}
	
\end{proof}

As a consequence, by using an optimal algorithm for $\ALG$, our mechanism is $\frac{1}{2}$-competitive with respect to the optimal social welfare.

%% file: appendix_knapsack_wbb.tex
\section{Appendix: Knapsack Double Auctions with weak budget-balance}
\label{appendix:knapsack_wbb}

We split the proof of Theorem \ref{Theorem:knapsack_wbb} in the two following lemmas.

\begin{lemma} \label{Lemma:Knapsack_DISC_IR_WBB}
	Mechanism \ref{mechanism_knapsack_wbb} for knapsack double auctions where no buyers demands more than half of the total capacity satisfies DWBB. Further, it is DSIC and IR for all buyers and sellers for any online adversarial order in which buyers and sellers are processed.
\end{lemma}

\begin{proof}
	Our mechanism is DWBB, as by construction, the mechanism consists of bilateral trades where an item is traded from one seller to one buyer. For any seller $j$ we have that $w_j^\ast \leq w_i^\ast$ for all buyers $i$, and hence $p_j \leq p_i$ for any buyer-seller pair $i,j$. The buyer pays $p_i$ to the mechanism and the seller receives $p_j$, so we get DWBB. IR follows naturally, DSIC from the fact that any agent is asked at most once in our mechanism.
\end{proof}

\begin{lemma}\label{Lemma:Knapsack_wbb_approximation}
	Mechanism \ref{mechanism_knapsack_wbb} for knapsack double auctions where no buyers demands more than half of the total capacity is $\frac{1}{5}$-competitive with respect to the optimal social welfare for any online adversarial order of buyers and sellers.
\end{lemma}

\begin{proof}
	The set of agents who receive an item $A$ depends on $\mathbf{v}$, so we denote by $A(\mathbf{v})$ the set $A$ under valuation profile $\mathbf{v}$. We want to compare $\Ex[\mathbf{v}]{\mathbf{v}(A(\mathbf{v}))}$ to $\Ex[\mathbf{v}]{\mathbf{v}(\OPT(\mathbf{v}))}$. To this end, again, we split the welfare of our algorithm into two parts, the base value and the surplus, and bound each quantity separately. The base value is thereby defined as follows: let agent $i$ receive an item in our mechanism, i.e. $i \in A$. Any buyer who gets an item has paid her agent-specific price for an item. Any seller who decided to keep her item was asked to keep it for her seller-specific price. The part of any agent $i$'s value which is below this price is denoted the base value. The surplus is the part of any agent $i$'s value above this threshold if it exists, otherwise it is zero. \\
	
	\textbf{Base Value}: Our base part of the social welfare is defined via the prices. Summing over all agents in $A(\mathbf{v})$, we can compute the following: 
	\begin{align*}
	\Ex[\mathbf{v}]{ \sum_{i \in A(\mathbf{v})} p_i   } = \frac{2}{5} \Ex[\widetilde{\mathbf{v}}]{\widetilde{\mathbf{v}}\left( \OPT (\widetilde{\mathbf{v}}) \right)} \cdot \Ex[\mathbf{v}]{\sum_{i \in A(\mathbf{v})} w_i^\ast} \geq \frac{2}{5} \Ex[\widetilde{\mathbf{v}}]{\widetilde{\mathbf{v}}\left( \OPT (\widetilde{\mathbf{v}}) \right)} \cdot \frac{1}{2} \Pr[\mathbf{v}]{\sum_{i \in A(\mathbf{v})} w_i^\ast \geq \frac{1}{2}} \enspace.
	\end{align*}
	
	\textbf{Surplus}: 
	We consider buyers and sellers separately and combine their respective contributions to the surplus afterwards. \\
	
	\textit{Sellers}: Note that any seller whose value exceeds her corresponding price can keep the item, so \[ \text{surplus}_j \geq \left( v_j - p_j \right)^+ \geq \left( v_j - p_j \right)^+ \cdot \mathds{1}_{\sum_{i' \in A\left((v_j', \mathbf{v}_{-j})\right)} w_{i'}^\ast \leq \frac{1}{2} } \cdot \mathds{1}_{j \in \OPT\left( (v_j , \mathbf{v}_{-j}') \right)} \enspace. \] Now, we use that $\mathbf{v}$ and $\mathbf{v}'$ are independent and identically distributed combined with linearity of expectation to get \[ \Ex[\mathbf{v}]{\text{surplus}_j } \geq \Ex[\mathbf{v}]{ \left( v_j - p_j \right)^+ } \geq \Ex[\mathbf{v}, \mathbf{v}']{ \left( v_j' - p_j \right)^+ \cdot \mathds{1}_{\sum_{i' \in A\left(\mathbf{v}\right)} w_{i'}^\ast \leq \frac{1}{2} } \cdot \mathds{1}_{j \in \OPT\left( \mathbf{v}' \right)} } \enspace. \]
	
	\textit{Buyers}: Concerning the buyers, note that buyer $i$ gets an item if buyer $i$'s value exceeds her price and if the sum of the weights of agents in $A$ does allow $i$ to be added. That is, denote by $A_{i-1}$ the set of accepted agents $A$ after processing buyer $i-1$. Then, we ensure $\sum_{i' \in A_{i-1} (\mathbf{v})} w_{i'}^\ast \leq 1 - w_i^\ast$. Note that $A_{i-1}$ does not depend on buyer $i$, so in particular $A_{i-1} (\mathbf{v}) = A_{i-1} \left((v_i',\mathbf{v}_{-i})\right)$. Further, a even stronger condition is that $\sum_{i' \in A \left((v_i',\mathbf{v}_{-i})\right)} w_{i'}^\ast \leq \frac{1}{2}$ as we did assume that $k \geq 2$ and $w_i \leq \frac{1}{2}$. Therefore, we can bound \[ \text{surplus}_i \geq \left( v_i - p_i \right)^+ \cdot \mathds{1}_{\sum_{i' \in A\left((v_i', \mathbf{v}_{-i})\right)} w_{i'}^\ast \leq \frac{1}{2} } \geq \left( v_i - p_i \right)^+ \cdot \mathds{1}_{\sum_{i' \in A\left((v_i', \mathbf{v}_{-i})\right)} w_{i'}^\ast \leq \frac{1}{2} } \cdot \mathds{1}_{i \in \OPT\left( (v_i , \mathbf{v}_{-i}') \right)}  \enspace. \] Again, using linearity of expectation as well as choosing $v_i'$ and $v_i$ to be independent and identically distributed, we get \[ \Ex[\mathbf{v}]{\text{surplus}_i} \geq \Ex[\mathbf{v}, \mathbf{v}']{\left( v_i' - p_i \right)^+ \cdot \mathds{1}_{\sum_{i' \in A\left(\mathbf{v}\right)} w_{i'}^\ast \leq \frac{1}{2} } \cdot \mathds{1}_{i \in \OPT\left( \mathbf{v}'\right)} } \enspace. \]
	
	\textit{Combination}: Summing over all buyers and sellers, we can combine the two bounds:
	\begin{align*}
	\Ex[\mathbf{v}]{\sum_{i \in B \cup S} \text{surplus}_i } & \geq \sum_{i \in B \cup S} \Ex[\mathbf{v}, \mathbf{v}']{\left( v_i' - p_i \right)^+ \cdot \mathds{1}_{\sum_{i' \in A\left(\mathbf{v}\right)} w_{i'}^\ast \leq \frac{1}{2} } \cdot \mathds{1}_{i \in \OPT\left( \mathbf{v}'\right)} } \\ & = \Pr[\mathbf{v}]{\sum_{i' \in A\left(\mathbf{v}\right)} w_{i'}^\ast \leq \frac{1}{2}} \cdot \Ex[\mathbf{v}']{ \sum_{i \in \OPT\left( \mathbf{v}'\right)} \left( v_i' - p_i \right)^+   } \\ & \geq \Pr[\mathbf{v}]{\sum_{i' \in A\left(\mathbf{v}\right)} w_{i'}^\ast \leq \frac{1}{2}} \cdot \left( \Ex[\mathbf{v}']{ \mathbf{v}' \left( \OPT(\mathbf{v}') \right)} - \Ex[\mathbf{v}']{ \sum_{i \in \OPT(\mathbf{v}')} p_i }  \right)  
	\end{align*}	
	
	To get the equality, note that $\mathbf{v}$ and $\mathbf{v}'$ are independent and the respective terms each only depend on one of the two. Now, in order to bound the sum of prices, we calculate \[ \Ex[\mathbf{v}']{ \sum_{i \in \OPT(\mathbf{v}')} p_i } = \frac{2}{5} \Ex[\widetilde{\mathbf{v}}]{ \widetilde{\mathbf{v}} \left( \OPT(\widetilde{\mathbf{v}}) \right)} \cdot \Ex[\mathbf{v}']{ \sum_{i \in \OPT(\mathbf{v}')} w_i^\ast  }  \enspace.  \] We use that we can bound $w_i^\ast = \max(w_i, \frac{1}{k}) \leq w_i + \frac{1}{k}$ on the buyers' side as well as $w_i^\ast = \frac{1}{k}$ for all sellers to get \[ \sum_{i \in \OPT(\mathbf{v}')} w_i^\ast \leq  \sum_{i \in \OPT(\mathbf{v}') \cap B} w_i + \sum_{i \in \OPT(\mathbf{v}')} \frac{1}{k} \leq 1+1 = 2 \] as the sum over the weights of all buyers in any feasible allocation is upper bounded by $1$ and further, we cannot allocate more than $k$ items in any feasible allocation, so $| \OPT(\mathbf{v}')| \leq k$. \\ Therefore, we can bound the overall surplus by \[ \Ex[\mathbf{v}]{\sum_{i \in B \cup S} \text{surplus}_i } \geq  \Pr[\mathbf{v}]{\sum_{i' \in A\left(\mathbf{v}\right)} w_{i'}^\ast \leq \frac{1}{2}} \cdot \left( 1 - \frac{4}{5} \right) \Ex[\widetilde{\mathbf{v}}]{ \widetilde{\mathbf{v}} \left( \OPT(\widetilde{\mathbf{v}}) \right)} \enspace. \]
	
	Summing the base value and the surplus proves our claim as we can exploit that $\mathbf{v}$, $\mathbf{v}'$ and $\widetilde{\mathbf{v}}$ are independent and identically distributed.
\end{proof}

In order to extend this to the general case when $w_i \in [0,1]$ instead of $w_i \leq 1/2$, we can run the following procedure: Split the set of buyers in those with $w_i \leq \frac{1}{2}$ and those with $w_i > \frac{1}{2}$. If we only consider the former buyers for trades, we know that our mechanism gives a $\frac{1}{5}$-fraction of the optimal social welfare. When restricting to the case of the latter buyers, the weights ensure that we can only allow at most one trade. Therefore, the setting simplifies to the matroid setting with the $1$-uniform matroid over the set of buyers for which we can extract half of the optimal social welfare using our mechanism from Section \ref{section:matroide_wbb}. Overall, estimating the expected welfare of each of the two options and selecting the better one will lead to a mechanism which always obtains at least a $\frac{1}{7}$-fraction of the optimal social welfare. Therefore, we can formulate Theorem \ref{Theorem:knapsack_wbb_unrestricted}.